\documentclass[12pt,a4paper]{article}
\usepackage{amsmath,amssymb,amsthm,eucal,cite}
\usepackage[top=3cm,right=3cm,left=2.5cm,bottom=3cm]{geometry}
\usepackage{tikz}
\usetikzlibrary{calc,matrix,patterns}
\usepackage{slashed}

\title{The power of the anomaly consistency condition for the Master Ward Identity:
Conservation of the non-Abelian gauge current}

\author{Michael D\"utsch%
\footnote{Institute f\"ur Theoretische Physik, Universit\"at G\"ottingen, 37077 G\"ottingen, Germany;\newline  
e-mail: michael.duetsch3@gmail.com}}

\date{}

\newcommand{\be}{\begin{equation}}   
\newcommand{\ee}{\end{equation}} 

\newcommand{\bt}{\beta}           
\newcommand{\Dl}{\Delta}          
\newcommand{\dl}{\delta}          
\newcommand{\eps}{\varepsilon}    
\newcommand{\ga}{\gamma}          
\newcommand{\ka}{\kappa}          
\newcommand{\La}{\Lambda}         
\newcommand{\la}{\lambda}         
\newcommand{\om}{\omega}          
\renewcommand{\th}{\theta}        
\newcommand{\vf}{\varphi}         

\newcommand{\bC}{\mathbb{C}}      
\newcommand{\bK}{\mathbb{K}}
\newcommand{\bM}{\mathbb{M}}      
\newcommand{\bN}{\mathbb{N}}      
\newcommand{\bR}{\mathbb{R}}      

\newcommand{\nS}{\mathbf{S}}      

\newcommand{\sC}{\mathcal{C}}     
\newcommand{\sD}{\mathcal{D}}     
\newcommand{\sF}{\mathcal{F}}     
\newcommand{\sL}{\mathcal{L}}     
\newcommand{\sM}{\mathcal{M}}     
\newcommand{\sO}{\mathcal{O}}     
\newcommand{\sP}{\mathcal{P}}     
\newcommand{\sR}{\mathcal{R}}     
\newcommand{\sS}{\mathcal{S}}     
\newcommand{\sV}{\mathcal{V}}     
\newcommand{\sW}{\mathcal{W}}

\newcommand{\s}{\mathfrak{s}}
\newcommand{\g}{\mathfrak{g}}
\newcommand{\h}{\mathfrak{h}}

\newcommand{\G}{\mathfrak{G}}
\newcommand{\LieG}{\mathrm{Lie}\G}
\newcommand{\LieGc}{\mathrm{Lie}\G_c}

\newcommand{\class}{\mathrm{class}}     
\newcommand{\comp}{\mathrm{c}}    
\newcommand{\id}{\mathrm{id}} 
\newcommand{\loc}{\mathrm{loc}}   
\newcommand{\ret}{\mathrm{class}}   
\newcommand{\reta}{\mathrm{ret}}   
\newcommand{\sym}{\mathrm{sym}}

\newcommand{\del}{\partial}       

\newcommand{\less}{\setminus}     
\newcommand{\ovl}{\overline}      
\newcommand{\ox}{\otimes}         
\newcommand{\bigox}{\bigotimes}         
\newcommand{\oxyox}{\otimes\cdots\otimes} 
\newcommand{\unl}{\underline}     
\newcommand{\wh}{\widehat}        
\newcommand{\x}{\times}           
\renewcommand{\.}{\cdot}          
\renewcommand{\:}{\colon}         

\DeclareMathOperator{\sgn}{sgn}    
\DeclareMathOperator{\supp}{supp}  
\newcommand{\tr}{\mathrm{tr}} 
\newcommand{\YM}{\mathrm{YM}}
\newcommand{\gf}{\mathrm{gf}} 
\newcommand{\gh}{\mathrm{gh}}
\newcommand{\photon}{\mathrm{photon}}
\newcommand{\inter}{\mathrm{int}}

\newcommand{\ldbrack}{[\mskip-2.5mu[} 
\newcommand{\rdbrack}{]\mskip-2.5mu]} 

\newcommand{\pw}[1]{\ldbrack#1\rdbrack} 
\newcommand{\set}[1]{\{\,#1\,\}}    
\newcommand{\Set}[1]{\biggl\{#1\biggr\}} 
\newcommand{\word}[1]{\quad\mbox{#1}\quad} 
\def\wick:#1:{\mathopen:#1\mathclose:} 

\newcommand{\fd}[2]{\frac{\dl#1}{\dl\vf(#2)}} 
\def\duo<#1,#2>{\langle#1,#2\rangle} 

\numberwithin{equation}{section}

\theoremstyle{plain}
\newtheorem{thm}{Theorem}[section]  
\newtheorem{prop}[thm]{Proposition} 
\newtheorem{lema}[thm]{Lemma}       
\newtheorem{corl}[thm]{Corollary}   

\theoremstyle{definition}

\newtheoremstyle{example}
   {\topsep}{\topsep}{\small}{0pt}%
   {\bfseries}{.}{ }{}
\theoremstyle{example}

\theoremstyle{remark}
\newtheorem{remk}[thm]{Remark}      

\DeclareRobustCommand{\qned}{\ifmmode
  \else \leavevmode\unskip\penalty9999 \hbox{}\nobreak\hfill \fi
  \quad\hbox{\qnedsymbol}}
\newcommand{\qnedsymbol}{$\boxminus$} 

\makeatletter
\renewcommand{\section}{\@startsection{section}{1}{\z@}%
                        {-3.5ex \@plus -1ex \@minus -.2ex}%
                        {2.3ex \@plus.2ex}%
                        {\normalfont\large\bfseries}}
\renewcommand{\subsection}{\@startsection{subsection}{2}{\z@}%
                        {-3.25ex \@plus -1ex \@minus -.2ex}%
                        {1.5ex \@plus .2ex}%
                        {\normalfont\normalsize\bfseries}}
\renewcommand{\subsubsection}{\@startsection{subsubsection}{3}{\z@}%
                        {-3.25ex \@plus -1ex \@minus -.2ex}%
                        {1.5ex \@plus .2ex}%
                        {\normalfont\normalsize\itshape}}
\renewcommand{\@dotsep}{200} 
\makeatother

\def\blue#1{{\color{black}#1}}

\begin{document}

\maketitle

\begin{abstract}
Extending local gauge tansformations in a suitable way to Faddeev-Popov ghost fields, one obtains a symmetry of
the total action, i.e., the Yang-Mills action plus a gauge fixing term (in a $\la$-gauge) plus the ghost action.
The anomalous Master Ward Identity (for this action and this extended, local gauge transformation) states
that the pertinent Noether current -- the interacting ``gauge current'' -- is conserved up to anomalies.

It is proved that, apart from terms being easily removable (by finite renormalization), all
possible anomalies are excluded by the consistency condition for the anomaly of the Master Ward Identity, recently derived 
in \cite{BDFR23}. 
\end{abstract}

\section{Introduction}

The main problem in the perturbative quantization of a classical field theory is the maintenance of classical symmetries, which is not 
always possible -- for certain symmetries there appear \emph{anomalies}. Theoretically the latter have been intensively investigated
(see e.g.~\cite{Bert}) and in several instances their appearance has been exprimentally confirmed. 
In particular, in a large class of applications, anomalies satisfy the \emph{Wess-Zumino consistency condition(s)} \cite{WZ71}. 

The Master Ward Identity (MWI) is a universal formulation of symmetries (see \cite{DB02,DF03} or \cite[Chap.~4]{D19}). 
In perturbative, classical field theory it is obtained by the pointwise multiplication of an 
arbitrary, interacting, local field \blue{$\bigl(Q(x)\bigr)^\class_{S_\inter}$ 
(where $S_\inter$ is the interaction, $Q$ is a polynomial in the basic field $\vf$ and
its partial derivatives $\del^a\vf$ and $\bigl(Q(x)\bigr)^\class_{S_\inter}$ is the 
corresponding classical interacting field, i.e.~it satisfies $\bigl(Q(x)\bigr)^\class_{S_\inter}\big\vert_{S_\inter=0}=Q(x)$)}
with the off-shell field equation, \blue{$\Bigl(\fd{(S_0+S_\inter)}{x}\Bigr)^\class_{S_\inter}=\fd{S_0}{x}$ 
(where $S_0$ is the free action); the resulting classical MWI can be written as
\be\label{eq:MWI-class}
\int dx\,\,q(x)\,\Bigl(Q(x)\cdot\fd{(S_0+S_\inter)}{x}\Bigr)^\class_{S_\inter}=
\int dx\,\,q(x)\,\bigl(Q(x)\bigr)^\class_{S_\inter}\cdot\fd{S_0}{x}\ ,
\ee
$q\in\sD(\bM)$ arbitrary; note that the r.h.s.~vanishes ``on-shell'', i.e., when the free field equation is valid.
On the l.h.s.~of \eqref{eq:MWI-class} it is used that 
\be\label{eq:prod-class-int-fields}
\bigl(Q(x)\bigr)^\class_{S_\inter}\cdot\Bigl(\fd{(S_0+S_\inter)}{x}\Bigr)^\class_{S_\inter}=
\Bigl(Q(x)\cdot\fd{(S_0+S_\inter)}{x}\Bigr)^\class_{S_\inter}\ ,
\ee
cf.~\eqref{eq:factorization}. However, for the corresponding interacting fields of perturbative Quantum Field Theory (pQFT), 
the pointwise product on the l.h.s.~of \eqref{eq:prod-class-int-fields}
does not exist (due to the distributional character of quantum fields), but the r.h.s.~of \eqref{eq:prod-class-int-fields} exists.
Hence, the MWI as written in \eqref{eq:MWI-class} is well-defined for the corresponding interacting fields of pQFT.
Requiring the MWI \eqref{eq:MWI-class} to hold in pQFT,} one obtains
a highly nontrivial renormalization condition, which cannot always be satisfied. \blue{However, one can prove the anomalous MWI (AMWI):} 
the main message of the pertinent theorem (see \cite[Thm.~7]{Brennecke08} or \cite[Thm.~4.3.1]{D19}, for the convenience of the reader we recall the relevant results in Thm.~\ref{thm:AMWI}) 
is that the MWI \blue{is valid for any renormalization prescription satisfisfying all other renormalization conditions, if one allows for} 
an additional term -- the anomalous term -- which is a local interacting field.
In addition, this theorem gives a lot of information about the structure of the anomalous term; however, there remained open the
question whether it satisfies a relation being analogous to the Wess-Zumino consistency condition -- without introducing 
antifields.
\footnote{In \cite[Prop.~4]{Hollands08}
Hollands gave a generalization of the AMWI to antifields and he derived a pertinent consistency condition for the anomaly
\cite[Prop.~5]{Hollands08}, but he wrote: ``These consistency conditions rely in an essential way upon the use
of the antifields, and this is the principal reason why we have introduced such fields in our construction''.}

Recently, this problem has been solved \cite{BDFR23}: for infinitesimal field transformations being \emph{affine} 
in the basic fields, such a relation has been derived directly from the definition of the anomalous term (that is, the AMWI)
without using antifields (Thm.~4.1 of that paper). In addition, the relation 
to the consistency condition for the anomaly in the 
BV formalism has been clarified. To wit, in the latter formalism, the anomaly consistency condition follows from the nilpotency of the
BV operator; this was already derived by Hollands \cite[Prop.~5]{Hollands08}.
\footnote{\blue{This reference treats only theories with closed gauge algebras.
Fr{\"o}b \cite{Fro19} has generalized Holland's results to theories with open gauge algebras (such as supersymmetric theories). 
Note that for this generalization antifields are indispensable:  they are required
 to close the gauge algebra off-shell and cannot be eliminated from the Lagrangian.}}
In \cite[Prop.~6.7]{BDFR23} it is proved that by suitable restriction of Hollands' anomaly consistency condition one obtains
the anomaly consistency condition of  \cite{BDFR23} -- a restriction eliminating the antifields from the (free and interacting) Lagrangian
and from the anomaly term.

In \cite{BDFR23} the derived anomaly consistency condition is called ``extended Wess-Zumino consistency condition'', because for
interactions being \emph{quadratic} in the basic fields, it reduces to the Wess-Zumino consistency condition \cite{WZ71}
\blue{(see also Remark \ref{rem:WZ})}.
However, it is not worked out in \cite{BDFR23}, what one gains by this extension in practice, that is, for the perturbative
quantization of a concrete model.

Hence, the main motivation for this paper is to study explicitly the restrictions on the anomaly of the Master Ward identity
coming from the anomaly consistency condition work out in \cite{BDFR23}, 
for a model for which this condition is non-trivial. The latter implies that the underlying symmetry
transformation must be non-Abelian. So we study local gauge transformations for massless, interacting Yang-Mills theories
in a $\la$-gauge à la `t Hooft (see e.g.~\cite{ADS98}). 
As it is well known, without additional fields, the quantized Yang-Mills field \blue{in a $\la$-gauge is problematic, because the 
gauge fixing term is not gauge invariant}  
-- the usual way out is to add  Faddeev-Popov ghosts (FP ghosts), we follow this path. 

A (famous) transformation leaving the total action (i.e., the Yang-Mills action plus a gauge fixing term plus the ghost action) 
invariant is the interacting 
BRST-tranformation; but, it is not affine, hence, it does not fit to our purposes.

However, a pure, local gauge transformation (i.e., the FP ghost are invariant) is not a symmetry of the considered model 
(see formula \eqref{eq:no-symmetry} for a precise formulation of this statement).
To remove this deficiency, we introduce ``extended'', local gauge transformations, transforming also the FP ghosts,
in such a way that this tranformation is both, a symmetry of the considered model and still an affine transformation.

The Noether current $J^\mu$ belonging to such an extended, local gauge transformation is called ``non-Abelian gauge 
current''. \blue{We study this current solely in a region of Minkowski space in which 
the test function $g$ switching the interaction is constant.}
Since the Yang-Mills Lagrangian is gauge invariant, the contributions to $J^\mu$ are coming from the
gauge fixing term and the ghost Lagrangian, the former contribution depends on the gauge fixing parameter $\la$.
Classically $J^\mu$ is conserved modulo the interacting field equations (Remark \ref{rem:class-dJ=0}). 
In pQFT the AMWI \eqref{eq:AMWI-explicit} states that conservation of the interacting gauge current, modulo the field 
equations for the underlying free fields (which imply the field equations for the interacting fields), may be violated 
in a well controlled way, that is, by anomalies fulfilling Thm.~\ref{thm:AMWI}. These possible anomalies satisfy also
the consistency condition derived in \cite{BDFR23}.
The most important result of this paper is that, \emph{apart from trivial anomalies} 
(i.e., anomalies which can be removed in an obvious way 
by a finite renormalization of the retarded product ($R$-product), or equivalently of the current $J^\mu$), 
\emph{all possible anomalies are excluded by this consistency condition}, \blue{provided that we are in a region in which 
the test function switching the interaction is constant.}

In the literature there are various proofs that massless Yang-Mills theories in a $\la$-gauge are anomaly free (i.e., all possible 
anomalies can be removed by finite, admissible renormalizations, see e.g.~\cite[Chap.~12-4]{IZ80}
or \cite{DHKS93,DHS95,ADS98}), even on curved spacetimes \cite{Hollands08},\blue{\cite{Fro19}}; 
typically these proofs are quite long and/or intricate. In this paper \blue{we prove only the conservation of the non-Abelian
gauge current -- we do not expect that the validty of this suffices for a pysically consistent construction 
of quantized Yang-Mills theories as e.g. given in \cite{DF99} (for QED), \cite{Hollands08} and \cite{Fro19} (however,
in all these works a proof  that the expectation values of the observables in the physical states are independent of the 
choice of the gauge fixing is missing).
Nevertheless, compared to the above mentioned proofs, our proof}
requires surprisingly little work \blue{-- this demonstrates the power of the anomaly consistency condition of \cite{BDFR23}.} 

\medskip

We use the approach to pQFT explained in the book \cite{D19} (which for the most parts relies on 
\cite{BF00,BrunettiDF09,DB02,DF99,DF01,DF01a,DF03,DF04,Brennecke08,DuetschFKR14}): fundamental building stones are that 
(classical and quantum) fields are functionals on the configuration space (``functional formalism''), quantization is achieved
by deformation quantization of the underlying free theory and interacting quantum fields are axiomatically defined and 
constructed by a version of Epstein-Glaser renormalization \cite{EG73}.
To avoid the listing of too many  references, we sometimes refer only to the book \cite{D19}, the original references can be found there.

In Sects.~\ref{sec:YM-basics}-\ref{sec:ccAMWI} we derive the above sketched results for massless Yang-Mills theories with
totally antisymmetric and non-vanishing structure constants. In Sect.~\ref{sec:generality} we discuss, for each of the main results
of Sects.~\ref{sec:YM-basics}-\ref{sec:ccAMWI},
which properties of the model are need in order that this result can be derived by the methods developed in this paper. 

In Sect.~\ref{sec:global-trafo} we study the \emph{global} transformation 
corresponding to the extended \emph{local} gauge transformation (treated in Sects.~\ref{sec:YM-basics}-\ref{sec:ccAMWI})
and the pertinent Noether current $j^\mu(g;x)$; the latter is closely related to the non-Abelian gauge current $J^\mu(x)$:
\footnote{\blue{It is the current $j^\mu(g;x)$ belonging to the global transformation that is usually studied in non-Abelian gauge theories.}}
\blue{in pQFT on-shell conservation of the corresponding interacting quantum currents is violated by the same anomaly term, 
provided that we are in a region in which the test function $g$
switching the interaction is constant (see \eqref{eq:dj=dJ}). A crucial advantage of the \emph{global} transformation} is that 
the underlying Lie group is \emph{compact} -- this does not hold for the local transformation.
By this compactness we gain that the Haar measure (on this group) is available; which makes possible the following procedure:
by a symmetrization of the $R$-product w.r.t.~this group, we
can reach that this product commutes with the global transformation (Prop.~\ref{pr:[Sa,R]=0}). This result 
is essentially used in Sects.~\ref{sec:YM-basics}-\ref{sec:ccAMWI}.
In addition, \blue{omitting the assumption that we are in a region in which the test function $g$ switching the interaction is constant,
this result (i.e., Prop.~\ref{pr:[Sa,R]=0})} yields a restriction of the possible anomaly of the conservation of the interacting
quantum Noether current $\bigl(j^\mu(g;x)\bigr)_{S_\inter(g)}$ belonging to the global
transformation (Prop.~\ref{pr:int-anomaly}). But it seems that a proof of the removability of this anomaly 
(by finite renormalizations of the $R$-product) requires quite a lot of additional work -- 
see the analogous proofs of the conservation of the 
electromagnetic current for spinor and scalar QED given in \cite{DF99} (or \cite[Chap.~5.2]{D19}) and \cite{DPR21}, respectively.


\section{Basics of massless Yang-Mills theories in the functional formalism}\label{sec:YM-basics}

\paragraph{The Lie algebra.}
Let $\LieG$ be a Lie algebra with totally antisymmetric
\footnote{In the approach to perturbative quantum gauge theories developed in the book \cite{Scharf01}
it turns out that such a theory is physically consistent only if the structure constants of the underlying Lie algebra are
totally antisymmetric.}
and non-vanishing (see \eqref{eq:f-nonvanishing} below) structure constants $f_{abc}$, e.g., $su(N)$. 
\blue{Total antisymmetry of the structure constants means that $\LieG$ is isomorphic to a direct sum of Abelian and simple compact Lie algebras (see e.g.~\cite[Chap.~3.5]{Scharf01}); the Abelian summands are excluded by the non-vanishing of the structure constants.}
Let $(T_a)_{a=1,\ldots,K}$ be a basis of the $\bR$-vector space $\LieG$; the pertinent structure constants 
are given in terms of the Lie bracket, $[\bullet,\bullet]:\LieG^{\ox 2}\to\LieG$, to wit
\be\label{eq:Lie-bracket-f}
[T_a,T_b]=f_{abc} T_c\ .
\ee
By the assumption that the structure constants are non-vanishing, we mean that for any
\footnote{The symbol ``$\sum_a$'' is mostly omitted, that is, repeated $\LieG$-indices are summed over.}  
$P=\sum_aP_aT_a\equiv P_aT_a\in\LieG$ the following conclusion holds:
\be\label{eq:f-nonvanishing}
f_{abc}P_c=0\,\,\,\forall a<b\quad\Rightarrow\quad P=0\ .
\ee
A  geometric interpretation of this assumption is given in the next section in \eqref{eq:f-not=0-geom}.

We will also use the trace, $\tr (\bullet\cdot\bullet):\LieG^{\ox 2}\to\bR$;
for $B=B_aT_a,\,C=C_aT_a\in\LieG$ (with $B_a,C_a\in\bR$) it is defined by
\be\label{eq:trace}
\tr (B\cdot C)\equiv\tr (BC):=\sum_a B_aC_a\in\bR\ ,\blue{\word{in particular}\tr(T_aT_b)=\dl_{ab}.} 
\ee
We also need the complexified Lie algebra: $\LieG_\bC:=\{c_a T_a\,\big\vert\,c_a\in\bC\}$.

\paragraph{Classical and quantum fields in the adjoint representation.}
Let $\bM$ be $4$-dimensional Minkowski space.
The gauge field $A^\mu$ and the FP ghost fields, $u$ and $\tilde u$, are in the adjoint representation: Writing $\rho$ for 
this representation and setting $t_a:=\rho(T_a)$ we have
\be\label{eq:adj-repr}
A^\mu(x) =A^\mu_a(x)\, t_a\ ,\quad u(x)=u_a(x)\,t_a\ ,\quad  \tilde u(x)=\tilde u_a(x)\,t_a\ ,
\word{with} (t_c)_{ab}:=f_{cba}=-f_{abc}
\ee
(the latter is the usual choice of a basis in $\rho(\LieG)$); to simplify the notations we mostly write $\LieG$ for $\rho(\LieG)$
(and similarly for $\LieG_\bC$). \blue{We assume that the basis of the Lie algebra, $(T_a)$, is normalized such that in the adjoint representation
the trace is given by $\frac{-1}{K}$ times the matrix trace, hence it holds that
\be\label{eq:tr-adjoint-rep}
\dl_{ab}=\tr(t_at_b)=\frac{-1}{K}\sum_{cd}(t_a)_{cd}(t_b)_{dc}=\frac{1}{K}\sum_{cd}f_{acd}f_{bcd}\ .
\ee}
In addition, we use the notation $\del A(x):=\del_\mu A^\mu_a(x)\,t_a$. 
The field strength tensor is
\be\label{eq:F-def}
F^{\mu\nu}(x):=\del^\mu A^\nu(x)-\del^\nu A^\mu(x)-\ka g(x) [A^\mu(x),A^\nu(x)]=F^{\mu\nu}_a(x)\,t_a\ ,
\quad g\in\sD(\bM,\bR)\ ,
\ee
where $\ka$ is the coupling constant; it is adiabatically switched off by multiplication with a test-function $g$.

We use the formalism in which classical and quantum fields are functionals on the (classical) configuration space $\sC$
(see \cite{D19} for an introduction to this formalism). In detail: Using the convention that $u$ is a real scalar field
and $\tilde u$ an anti-real scalar field (i.e, it takes valus in $i\bR$), the configuration space is%
\footnote{We recall that $K:=\dim\LieG$. \blue{The formulation of fermionic fields 
in the functional formalism given in \cite[Chap.~5.1]{D19} (which we use here)
may be viewed as a mathematically more elementary version of the one given in \cite{Rejzner11}.}}
\begin{align}\label{eq:configurations}
\sC:=C^\infty(\bM,\bR^{4K})&\x \oplus_{n=1}^\infty C^\infty(\bM,(\bR\x i\bR)^{K})^{\x n}\nonumber\\
&\simeq C^\infty(\bM,\LieG^{\x 4})\x \oplus_{n=1}^\infty C^\infty(\bM,\LieG\x i\LieG)^{\x n}\ ,
\end{align}
where, for $(h_a)\in C^\infty(\bM,\bR^{4K}),\, (v_a)\in C^\infty(\bM,\bR^{K}),\,
(\tilde v_a)\in C^\infty(\bM,(i\bR)^{K})$, 
we identify $(h^\mu_a)$ with $h^\mu(x):=h^\mu_a(x)\,t_a\in C^\infty(\bM,\LieG)$ and $(v_a),\,(\tilde v_a)$ with
$v(x):=v_a(x)t_a\in C^\infty(\bM,\LieG),\,\tilde v(x):=\tilde v_a(x) t_a\in C^\infty(\bM,i\LieG)$, respectively; 
see \cite[Chap.~5.1]{D19} for the configuration space of fermionic fields. (The reason 
for the direct sum $\oplus_{n=1}^\infty C^\infty(\bM,\ldots)^{\x n}$ in the fermionic configuration space is that
the classical product of fermionic fields is the wedge product, see \eqref{eq:wedge}.)

With this, the basic fields $A^\mu(x)$, $F^{\mu\nu}(x)$, $u(x)$ and $\tilde u(x)$ are the following evaluation functionals:%
\footnote{Besides the Lie bracket, we use the $[\,\,]$-brackets also for the application of a field (i.e., a functional) to a configuration
(i.e., a triple of smooth functions, see \eqref{eq:configurations}), and also for the commutator of functional differential operators, see
Lemma \ref{lem:dX=repr}.}
\begin{align*}
&A^\mu(x)[h,v,\tilde v]:=h^\mu(x)=h^\mu_a(x)\,t_a\in\LieG\ ,\\
&F^{\mu\nu}(x)[h,v,\tilde v]:=\del^\mu h^\nu(x)-\del^\nu h^\mu(x)-\ka g(x) [h^\mu(x),h^\nu(x)]\in\LieG\ ,\\
&u(x)[h,v,\tilde v]:=v(x)=v_a(x)\,t_a\in\LieG\ ,\\
&\tilde u(x)[h,v,\tilde v]:=\tilde v(x)=\tilde v_a(x)\,t_a\in i\LieG\ .
\end{align*}

The space of (classical and) quantum \emph{fields} $\sF=\sF_\YM\ox\sF_\gh$ is defined as
the vector space of functionals $F \equiv F(A,u,\tilde u): \sC \to \bC$ or $\sC \to \LieG_\bC$ of the form
\be
F(A,u,\tilde u) = f_0+\sum_{n=1}^N \sum_{j_1,\dots,j_n}\int dx_1 \cdots dx_n \,\,\vf_{j_1}(x_1) \cdots \vf_{j_n}(x_n)\,
f_n^{j_1,\dots,j_n}(x_1,\dots,x_n)
\label{eq:fields} 
\ee
with  $N < \infty$, where 
\be\label{eq:fields-1}
F(A^\mu,u,\tilde u)[h,v,\tilde v]:=F(h,v,\tilde v) \word{and} \vf_j\in\{A^\mu_a,u_a,\tilde u_a\}\ ,
\ee
i.e., the index $j$ of $\vf_j$ labels the components of the various basic fields; in particular, the sum over $j$
contains the sum over the $\LieG$-index $a$ of the basic fields $\vf_j$. In addition, the product of the $\vf_j$'s
is the  \emph{classical} product, that is, the pointwise product of functionals for bosonic fields and the 
wedge product of functionals for fermionic fields, e.g.,
\begin{align}\label{eq:wedge}
\vf_{j_1}(x_1) \cdots &\vf_{j_n}(x_n)= A^{\mu_1}_{a_1}(x_1)\cdots A^{\mu_k}_{a_k}(x_k)\\
&\ox u_{a_{k+1}}(x_{k+1})\wedge\cdots\wedge u_{a_{k+s}}(x_{k+s})\wedge
\tilde u_{a_{k+s+1}}(x_{k+s+1})\wedge\cdots\wedge\tilde u_{a_n}(x_n)\ ,\nonumber
\end{align}
where $k+s\leq n$.
In \eqref{eq:fields} the $f_n^{\dots}$'s are as follows: $f_0 \in \bC$ (or  $f_0 \in \LieG_\bC$) is a
constant and, for $n \geq 1$, $f_n^{j_1,\ldots,j_n}$ is a distribution with compact support, i.e., $f_n^{j_1,\ldots,j_n}\in\sD'(\bM^n,\bC)$
(or $\sD'(\bM^n,\LieG_\bC)$). In addition, each $f_n^{j_1,\ldots,j_n}$ is required to satisfy a certain condition on its
wave front set, which ensure the existence of pointwise products of distributions appearing in the star product, see \cite{DF01}.

The space of ``on-shell'' fields is
\be\label{eq:on-shell}
\sF_0:=\{F_0:=F\vert_{\sC_0}\,\big\vert\,F\in\sF\} \ ,
\ee
where $\sC_0$ is the subspace of $\sC$ of solutions of the free field equations.

The space $\sF_\loc$ of \emph{local fields} is the
following subspace of $\sF$: Let $\sP$ be the space of polynomials in the variables
$\{\del^rA^\mu_a,\del^ru_a,\del^r\tilde u_a\,|\,r \in \bN^4\}$ with real coefficients (``field polynomials''; again, 
the $\del^r u$'s and $\del^r\tilde u$'s are multiplied by the wedge product); then
\begin{equation}
\sF_\loc := \Set{\sum_{i=1}^I \int dx\,\, B_i(x)\, g_i(x)
\,\Big\vert\, B_i \in \sP, \ g_i \in \sD(\bM,\bC)\,\,\text{or}\,\,g_i \in \sD(\bM,\LieG_\bC), \ I < \infty}.
\label{eq:local-fields} 
\end{equation}

The $*$-operation is defined by complex conjugation, that is, 
\be 
{A^\mu_a}^*(x)=A^\mu_a(x)\word{and} u_a^*(x)=u_a(x)\ ,\quad \tilde u_a^*(x)=-\tilde u_a(x)\ ,
\ee
since $\tilde u_a$ is anti-real; and the $*$-conjugated field of $F\in\sF$ given in \eqref{eq:fields} is
\be
F^* = \ovl{f_0}+\sum_n\sum_{j_1,\dots,j_n}\int dx_1 \cdots dx_n \,\,\vf_{j_n}^*(x_n) \cdots \vf_{j_1}^*(x_1)\,
\ovl{f_n^{j_1,\dots,j_n}(x_1,\dots,x_n)}\ .\label{eq:F*}
\ee

The vacuum state of the free theory is
\be\label{eq:vacuum}
\sF\ni F\mapsto\om_0(F):=F[0]=f_0\in\bC\ ,
\ee
by using \eqref{eq:fields}-\eqref{eq:fields-1}.

\paragraph{Deformation quantization of the free theory.}
Using deformation quantization for the quantization of the underlying free theory
(with the gauge field in a $\la$-gauge à la `t Hooft), we define the star-product by
\footnote{The lower indices `l' and `r' signify whether the pertinent functional derivative w.r.t.~a fermionic field is acting
from the left- or right-hand side, respectively. If such an index is missing, we mean the one acting from the left-hand side.}
\begin{align}\label{eq:star-product}
&F\star G=\sM\circ e^{\hbar\sD}(F\ox G) \word{for $F,G\in\sF$, with}\\
\sD:=\int dx\,dy\,\,& \Bigl(-D_\la^{\mu\nu,+}(x-y)\,\frac{\dl}{\dl A^\mu_d(x)}\ox\frac{\dl}{\dl A_d^\nu(y)}\nonumber\\
&+D^+(x-y)\Bigl(\frac{\dl_r}{\dl u_d(x)}\ox\frac{\dl_l}{\dl \tilde u_d(y)}
-\frac{\dl_r}{\dl \tilde u_d(x)}\ox\frac{\dl_l}{\dl u_d(y)}\Bigr)\Bigr),\nonumber
\end{align}
where $\sM$ is the classical product,
\footnote{We recall that for fermionic fields this is the wedge product \eqref{eq:wedge}.} 
$\sM(F\ox G)=F\cdot G$; 
for details see \cite{DF01,DF01a,DF04} and \cite[Chaps. 2 and 5.1]{D19}. 
In addition, $D^+$ is the massless Wightman 2-point function and 
\be\label{eq:dipol-0}
D_\la^{\mu\nu,+}(z):=g^{\mu\nu}\, D^+(z)+\frac{1-\la}{\la}\,(\del^\mu\del^\nu E)^+(z)\ ,\quad
\blue{g:=\mathrm{diag}(+,-,-,-)}\ ,
\ee
where $E$ is the ``dipole distribution'', which is defined by
\be\label{eq:dipol-1}
\square^2 E=0\ ,\quad (\del_0^{\,\,k} E)(0,\vec x)=0\quad k=0,1,2\ ,\quad  (\del_0^{\,\,3} E)(0,\vec x)=-\dl(\vec x)\ ,
\ee
and $(\del^\mu\del^\nu E)^+$ is the positive frequency part of $i(\del^\mu\del^\nu E)$, for details see e.g.~\cite{ADS98}.
(Note that $E^+$ is ill-defined.) Explicitly, the solution of \eqref{eq:dipol-1} reads
\footnote{\blue{We use the following notations and conventions: $p^2:=p^\mu p_\mu$, $px:=p^\mu x_\mu$ and for the Fourier
transformation $\widehat f(p):=\frac{1}{(2\pi)^2}\int dx\,\,f(x)\,e^{ipx}$.}}
\be\label{eq:dipol-2}
E(x)=\frac{-i}{(2\pi)^3}\int d^4p\,\,\sgn(p^0)\,\dl'(p^2)\,e^{-ipx}=
\frac{-1}{8\pi}\,\sgn(x^0)\,\Theta(x^2)\ ,
\ee
and it satisfies
\be
(\del_\mu\del_\nu E)(0,\vec x)=0\ ,\quad (\del_\mu\del_\nu\del_j E)(0,\vec x)=0\quad j=1,2,3\ .
\ee
From these results and the well-known relations
\be
D(0,\vec x)=0\ ,\quad (\del_jD)(0,\vec x)=0\,\,\,j=1,2,3\ ,(\del_0D)(0,\vec x)=-\dl(\vec x)
\ee
we see that $D_\la^{\mu\nu}(x):=-i\bigl(D_\la^{\mu\nu,+}(x)-D_\la^{\mu\nu,+}(-x)\bigr)$ 
\blue{(i.e.,  the ``commutator function'' for the gauge field in a $\la$-gauge)} fulfils the initial conditions
\begin{align}\label{eq:D-la(0,x)}
D_\la^{\mu\nu}(0,\vec x)={}&0\ ,\quad(\del_jD_\la^{\mu\nu})(0,\vec x)=0\quad j=1,2,3\ ,\\
(\del_0D_\la^{\mu\nu})(0,\vec x)={}&\begin{cases}-\frac{1}{\la}\,\dl(\vec x)\quad\text{if $(\mu,\nu)=(0,0)$}\\
-g^{\mu\nu}\,\dl(\vec x)\quad\text{if $(\mu,\nu)\not= (0,0)$}\end{cases}
=-g^{\mu\nu}\,\dl(\vec x)\Bigl(1+\frac{1-\la}{\la}\dl_{\nu 0}\Bigr)\nonumber
\end{align}

The star product induces a well-defined product on the space $\sF_0$ of on-shell fields 
(which is also denoted by `$\star$') by the definition
\be\label{eq:on-shell-star-product}
F_0\star G_0:=(F\star G)_0\ ,
\ee
because $D^+$ and $D_\la^{\mu\nu,+}$ are solutions of the pertinent free field equation.

\blue{The connection to the Fock--Krein
\footnote{\blue{The algebra of on-shell spin $1$ gauge fields cannot be 
nontrivially represented on a pre Hilbert space such that the $*$-operation is respected, hence one usually works 
with an inner product space or a Krein space.}}
space formalism is given by the following result 
(see \cite{DF01,DF01a} or \cite[Thm.~2.6.3]{D19}): there is an injective algebra bi-homomorphism 
$\Phi$ from $(\sF_0,\,\cdot\,,\star)$ ("$\,\cdot\,$" stands for the classical product) to the space of linear operators on 
Fock--Krein space with the normally ordered product and the operator product, explicitly the intertwining relations read
\be
\Phi(F_0 \. G_0)=\wick:\Phi(F_0)\,\Phi(G_0):\word{and}\Phi(F_0\star G_0)=\Phi(F_0)\,\Phi(G_0)\ ;
\ee
in addition, $\Phi$ respects the $*$-operation:
$\big\langle\Psi_1,\Phi(F_0^*)\Psi_2\big\rangle=\big\langle\Phi(F_0)\Psi_1,\Psi_2\big\rangle$
for $\Psi_1,\Psi_2$ in the domain of $\Phi(F_0)$ or $\Phi(F_0^*)$, respectively.}

\paragraph{Classical field equations.}
The Yang-Mills (YM), gauge fixing (gf) and FP ghost (gh) Lagrangian, in a $\la$-gauge (where $\la\in\bR\less\{0\}$), read
\be\label{eq:Lagrangians}
L_\YM:=-\tfrac{1}4\,\tr( F^{\mu\nu} F_{\mu\nu})\ ,\quad L_\gf:= -\tfrac{\la}2\,\tr\bigl((\del A)(\del A)\bigr)\ ,
\quad L_\gh:=\tr\bigl((\del_\mu\tilde u)(D^\mu u)\bigr)\ ,
\ee
and the total action is
\be
S:=S_\YM+S_\gf+S_\gh\ ,\quad S_{\cdots}:=\int dx\,\,L_{\cdots} (x)\ ,
\ee
with the covariant derivative 
\be\label{eq:cov-derivat}
D^\mu :=\del^\mu+\ka g [\bullet,A^\mu]\ ,\word{that is,} 
\blue{D^\mu_{ab}(x):=\dl_{ab}\del_x^\mu+\ka g(x)\,f_{abc}A^\mu_c(x) \ ,}
\ee
in particular
\be
D^\mu u(x)=\blue{D^\mu_{ab}(x)u_b(x)\,t_a=}\bigl(\del^\mu u_a(x)+\ka g(x)\,f_{abc}\,u_b(x)\,A^\mu_c(x)\bigr)t_a\ .
\ee

Assuming that $g=1$ in a neighbourhood of $x$, the field equations $\frac{\dl S}{\dl A_{a,\mu}(x)}=0$, $\frac{\dl S}{\dl u_a(x)}=0$ 
and $\frac{\dl S}{\dl \tilde u_a(x)}=0$ can be written as:
\footnote{Note that $[\del^\mu\tilde u(x),u(x)]:=\del^\mu\tilde u_b(x)\wedge u_c(x) [t_b,t_c]=
f_{abc}\,\del^\mu\tilde u_b(x)\wedge u_c(x)\,t_a$;
hence it holds that $[\del^\mu\tilde u(x),u(x)]=[u(x),\del^\mu\tilde u(x)]$ (due to 
$\del^\mu\tilde u_b(x)\wedge u_c(x)=-u_c(x)\wedge\del^\mu\tilde u_b(x)$).}
\begin{align}
\square A^\mu(x)-&(1-\la)\del^\mu(\del A)(x)
=\ka \Bigl(\del_\nu [A^\nu(x),A^\mu(x)]+[A_\nu(x),(\del^\nu A^\mu(x)-\del^\mu A^\nu(x)]\nonumber\\
&-[\del^\mu\tilde u(x),u(x)] \Bigr)
+\ka^2 [A_\nu(x),[A^\mu(x),A^\nu(x)]]\ ,\label{eq:FE-A}\\
0=&D^\mu\del_\mu  \tilde u(x)=\square \tilde u(x)+\ka [\del_\mu\tilde u(x),A^\mu(x)]\ ,\label{eq:FE-tilde-u}\\
0=&-\del_\mu D^\mu u(x)=-\square u(x)-\ka\,\del_\mu[u(x),A^\mu(x)]\label{eq:FE-u}\ .
\end{align}
On-shell -- that is, restricted to configurations solving the free field equations -- the field equations are valid for 
classical, perturbative interacting fields (see \eqref{eq:class-int-field}); and, if the renormalization condition 
``off-shell field equation'' (given in Sect.~\ref{sec:AMWI}) is satisfied, this holds
also for the (partly composite) interacting fields of perturbative QFT \cite{D19}.


 \section{Extended, infinitesimal local gauge transformation with compact support}\label{sec:gauge-trafo}
 
 \paragraph{$\LieG$-rotation.}
Considering \emph{infinitesimal} transformations, we will use the notion of a ``$\LieG$-rotation'' in the adjoint representation $\rho$. 
Under $t_a\in\rho(\LieG)$ an element $P=P_b\,t_b$ of $\rho(\LieG)$ is transformed as follows:
\be\label{eq:LieG-Rot}
\s_a\,:\begin{cases}\rho(\LieG)\longrightarrow\rho(\LieG)\\
P\mapsto \s_a(P):=(t_aP)_b\,t_b\equiv(t_a)_{bc}P_c\,t_b=f_{acb}P_c\,t_b=[t_a,P]\ .\end{cases}
\ee
By the Jacobi identity, $\s_a$ is a derivation w.r.t.~the Lie bracket:
\be
\s_a([P,Q])=[\s_a(P),Q]+[P,\s_a(Q)]\ ,\quad P,Q\in\rho(\LieG)\ .
\ee

With this, the assumption that the structure constants are non-vanishing \eqref{eq:f-nonvanishing}, can equivalently be written as
\be\label{eq:f-not=0-geom}
 \s_a(P)=0\,\,\,\forall a\quad\Rightarrow\quad P=0\ .
\ee
 
A ``$\LieG$-rotation'' for \emph{fields} being in the adjoint representation is defined as follows: Using the notation ``$s_a$''
for the transformation given by $t_a$, a  basic field $\vf(x)=\vf_b(x)\,t_b$ transforms as in \eqref{eq:LieG-Rot},
that is,
\be\label{eq:LieG-rot}
s_a\,:\,\vf(x)\mapsto s_a\bigl(\vf(x)\bigr):=
f_{acb}\vf_c(x)\,t_b=[t_a,\vf(x)]\ .
\ee
The transformation of a field $F\in\sF$ being composed of basic fields $\vf_j(x)=\vf_{j,b}(x)\,t_b$ is defined 
by the requirement that $s_a:\sF\to\sF$ is a derivation w.r.t.~the 
classical product (i.e., the pointwise or wedge product of functionals), hence
\be\label{eq:LieG-rot-1}
F\mapsto s_a(F):=\sum_{j,b}\int dx\,\,[t_a,\vf_j(x)]_b\,\frac{\dl F}{\dl \vf_{j,b}(x)}=
-f_{abc}\int dx\,\,\vf_{j,c}(x)\,\frac{\dl F}{\dl \vf_{j,b}(x)}\ .
\ee
\blue{Note that
\be\label{eq:sa-fd}
s_a\bigl(\del^\mu P(x)\bigr)=\del^\mu_x s_a\bigl(P(x)\bigr)\word{and}
s_a\Bigl(\frac{\dl F}{\dl\vf_{k,b}(x)}\Bigr)=\frac{\dl s_a(F)}{\dl\vf_{k,b}(x)}-f_{abc}\,\frac{\dl F}{\dl\vf_{k,c}(x)}\ ,
\ee
for $P\in\sP$. Motivated by $s_a\bigl(\vf_b(x)\bigr)=-f_{abc}\,\vf_c(x)$, $\LieG$-covariant fields are defined as follows: 
\be
F_b\in\sF\word{is a $\LieG$-vector iff} s_a(F_b)=-f_{abc}\,F_c\,\,\,(\Leftrightarrow s_a(F)=[t_a,F]\,\,\text{for $F:=F_bt_b$,})
\label{eq:LieG-vector}\\
\ee
\be
G_{bc}\in\sF\word{is a $\LieG$-tensor of 2nd rank iff} s_a(G_{bc})=-f_{abd}\,G_{dc}-f_{acd}\,G_{bd}\ .\label{eq:LieG-tensor}
\ee}
For example, using the notations introduced in \eqref{eq:fields}, we obtain
\begin{align}
s_a\Bigl(\int dx_1&\cdots dx_4\,\,f^b_4(x_1,\ldots,x_4)\,[\vf_{j_1}(x_1),\vf_{j_2}(x_2)]_b\,
\tr(\vf_{j_3}(x_3)\cdot\vf_{j_4}(x_4))\Bigr)\\
=&\int dx_1\cdots dx_4\,\,f^b_4(x_1,\ldots,x_4)\nonumber\\
&\cdot\Bigl(\bigl([[t_a,\vf_{j_1}(x_1)],\vf_{j_2}(x_2)]_b
+[\vf_{j_1}(x_1),[t_a,\vf_{j_2}(x_2)]]_b\bigr)\,\tr(\vf_{j_3}(x_3)\cdot\vf_{j_4}(x_4))\nonumber\\
&+[\vf_{j_1}(x_1),\vf_{j_2}(x_2)]_b\,\bigl(\tr([t_a,\vf_{j_3}(x_3)]\cdot\vf_{j_4}(x_4))
+\tr(\vf_{j_3}(x_3)\cdot[t_a,\vf_{j_4}(x_4)])\bigr)\Bigr)\nonumber\\
=&\blue{\int dx_1\cdots dx_4\,\,f^b_4(x_1,\ldots,x_4)\,[t_a,[\vf_{j_1}(x_1),\vf_{j_2}(x_2)]]_b
\,\tr(\vf_{j_3}(x_3)\cdot\vf_{j_4}(x_4))}\ ,\nonumber
\end{align}
\blue{by using the Jacobi identity for the double Lie bracket terms} and
\be\label{eq:sa(tr)}
s_a\Bigl(\tr(\vf_{j_1}(x)\cdot\vf_{j_2}(y))\Bigr)=\tr([t_a,\vf_{j_1}(x)]\cdot\vf_{j_2}(y))+\tr(\vf_{j_1}(x)\cdot[t_a,\vf_{j_2}(y)])=0\ ,
\ee
the latter is due to the total antisymmetry of the structure constants. Hence it holds that $s_a\bigl(L_\gf(x)\bigr)=0$. Using
additionally the Jacobi identity for the Lie bracket, we also obtain
\begin{align}\label{eq:sa(Lgh)}
s_a(L_\gh)=&\,\tr\bigl(\del_\mu [t_a,\tilde u]\cdot D^\mu u\bigr)+\tr\Bigl(\del\tilde u\cdot\Bigl(\del^\mu [t_a,u]+
\ka g\bigl([[t_a,u],A^\mu]+[u,[t_a,A^\mu]]\bigr)\Bigr)\Bigr)\nonumber\\
=&\,\tr\bigl([t_a,\del_\mu \tilde u]\cdot D^\mu u\bigr)+\tr\bigl(\del_\mu \tilde u\cdot [t_a,D^\mu u]\bigr)=0
\end{align}
and 
\begin{align}\label{eq:sa(LYM)}
s_a(L_\YM)=&-\tfrac{1}2\,\tr\Bigl(\Bigl(\del^\mu [t_a,A^\nu]-\del^\nu [t_a,A^\mu]-\ka g\bigl([[t_a,A^\mu],A^\nu]+
[A^\mu,[t_a,A^\nu]]\bigr)\Bigr)\cdot F_{\mu\nu}\Bigr)\nonumber\\
=&\,-\tfrac{1}2\,\tr\bigl([t_a,F^{\mu\nu}]\cdot F_{\mu\nu}\bigr)=0
\end{align}
(using $f_{abc}F_b^{\mu\nu}\cdot F_{c\,\mu\nu}=0$ in the last step).

\paragraph{\blue{Extended}, local gauge transformations with compact support.}
Since we want to study \emph{local} transformations, with \emph{compact support}, we consider the Lie algebra
 $$
 \LieGc:=\sD(\bM,\LieG)\ ,\word{where} [X,Y](x):=[X(x),Y(x)]\word{for} X,Y\in\LieGc,\,x\in\bM\ .
 $$
 Note that $X\in\LieGc$ is of the form $X(x)=X_a(x)\,t_a$ with $X_a\in\sD(\bM,\bR)$.
 
In the remainder of this sect.~and in sects.~\ref{sec:AMWI}-\ref{sec:generality}
we solely study pairs $(X,g)$ (where $X\in\LieGc$ and $g$ is the test-function switching $\ka$) satisfying
\be\label{eq:g=1}
g\big\vert_{\supp X}=1\ .
\ee

As it is very well known, the action on the configuration space of a local gauge transformation given by $X\in\LieGc$ is
for $h$ (i.e., the configuration of the gauge field $A$)
not only a local $\LieGc$-rotation (multiplied by $(-\ka)$); additionally it contains an infinitesimal field shift 
(see Footnote \ref{fn:field-shift}):
\footnote{In order to agree with the conventions used in \cite{BDFR23}, we write $hX$ instead of $Xh$ \blue{and analogously for 
$v$ and $\tilde v$}.}
\blue{\begin{align}\label{eq:hX}
&\sC\x\LieGc\ni ((h,v,\tilde v),X)\mapsto (hX,vX,\tilde vX)\word{with}\\
& (hX)^\mu(x):=-\bigl(\del^\mu X(x)+\ka [X(x),h^\mu(x)]\bigr)\ ;\nonumber
\end{align}
and we extend this transformation to the ghost fields by a local $\LieGc$-rotation (multiplied by $(-\ka)$):
\be
(vX)(x):=-\ka[X(x),v(x)]\,,\quad (\tilde v X)(x):=-\ka[X(x),\tilde v(x)]\ .
\ee
We point out that this is the infinitesimal version of an \emph{affine} configuration transformation with compact support, i.e., 
it fits in the framework studied in \cite{BDFR21,BDFR23}.}

\blue{A \emph{Lie algebra representation} of $\LieGc$ by linear maps on $\sF_\loc$ is given by
\begin{align}\label{eq:del_X}
\LieGc\ni  X&\mapsto \del_X:=\del_X^A+\del_X^u+\del_X^{\tilde u}:\sF_\loc\to\sF_\loc\ ,\word{with}\\
\del_X^A&:=-\sum_a\int dx\,\,\Bigl(\del^\mu X_a(x)+\ka [X(x),A^\mu(x)]_a\Bigr)\,\frac{\dl}{\dl A_a^\mu(x)}\nonumber\\
&=:-\int dx\,\,\Bigl(\del^\mu X(x)+\ka [X(x),A^\mu(x)]\Bigr)\,\frac{\dl}{\dl A^\mu (x)}\ ,\nonumber\\
\del_X^u&:=-\ka\sum_a\int dx\,\, [X(x),u(x)]_a\,\frac{\dl}{\dl u_a(x)}=:-\ka\int dx\,\,[X(x),u(x)]\,\frac{\dl}{\dl u(x)}\ ,
\nonumber\\
\del_X^{\tilde u}&:=-\ka\sum_a\int dx\,\, [X(x),\tilde u(x)]_a\,\frac{\dl}{\dl \tilde u_a(x)}=:
-\ka\int dx\,\,[X(x),\tilde u(x)]\,\frac{\dl}{\dl \tilde u(x)}\ ,\nonumber
\end{align}
where the respective last expressions are shorthand notations.
\footnote{For $\del^A_X$, the corresponding (infinitesimal) gauge transformation for the free theory (i.e., $\ka=0$), can be understood as 
infinitesimal field shift:
$$
(\del_X^A\big\vert_{\ka=0}F)[h,v,\tilde v]=\frac{d}{d\eps}\Big\vert_{\eps=0}F[h-\eps\del X,v,\tilde v]\ .
$$\label{fn:field-shift}} 
In the definition of $\del^A_X$, the infinitesimal field shift and the $\LieGc$-rotation can be summarized by the covariant derivative \eqref{eq:cov-derivat}:
\be\label{eq:del_X-D}
\del_X^A=-\int dx\,\,\bigl(D^\mu_{ba}(x)\,X_a(x)\bigr)\,\frac{\dl}{\dl A_b^\mu(x)}=
\int dx\,X_a(x)\,D^\mu_{ab}(x)\,\frac{\dl}{\dl A_b^\mu(x)}\ .
\ee
In particular, we obtain
\begin{align}\label{eq:del_X-basic}
\del_X A^\mu(x)=&-\bigl(\del^\mu X(x)+\ka [X(x),A^\mu(x)]\bigr)=-D^\mu(x)\,X(x)\ ,\nonumber\\
\del_X u(x)=&-\ka [X(x),u(x)]\ ,\quad \del_X \tilde u(x)=-\ka [X(x),\tilde u(x)]\ .
\end{align}
Note that
\be
\del_X A^\mu(x)[h,v,\tilde v]=\bigl((hX)^\mu(x),v(x),\tilde v(x)\bigr)\ ,\quad
\del_X u(x)[h,v,\tilde v]=\bigl(h(x),(vX)(x),\tilde v(x)\bigr)
\ee
and analogously for $\del_X \tilde u(x)[h,v,\tilde v]$ -- in agreement with the formalism developed in \cite{BDFR21,BDFR23}.}

\blue{From the definition of $\del_X$ \eqref{eq:del_X}, we see that it is a functional differential operator, hence it}
is a derivation on the space $\sP$ of polynomials in the 
basic fields $A^\mu,\,u,\,\tilde u$ and its partial derivatives, explicitly:
$$
\del_X\bigl(P(x)\cdot R(x)\bigr)=\bigl(\del_X P(x)\bigr)\cdot R(x)+P(x)\cdot\bigl(\del_X R(x)\bigr)\word{for $P,R\in\sP$.}
$$

With the assumption \eqref{eq:g=1} and by using that 
\be\label{eq:dXdA=ddXA}
\del_X(\del^\mu A^\nu)(x)=\del^\mu(\del_X A^\nu(x))
\ee
(which follows from \eqref{eq:del_X}) we obtain
\begin{align}
\del_X F^{\mu\nu}(x)=&-\{\del^\mu\del^\nu X(x)+\ka\del^\mu [X(x),A^\nu(x)]
-\ka\bigl[(\del^\mu X(x)+\ka [X(x),A^\mu(x)]),A^\nu(x)\bigr]\}\nonumber\\
&+\{\mu\leftrightarrow\nu\}\nonumber\\
=&-\ka\,[X(x),F^{\mu\nu}(x)]\ ,
\end{align}
by using the Jacobi identity for the Lie bracket in the last step; that is, the transformation of $F^{\mu\nu}$ is 
simply a local $\LieGc$-rotation. A glance at \eqref{eq:sa(LYM)} shows that
\be\label{eq:dX(LYM)}
\del_X\,L_\YM(x)=0\ .
\ee

\medskip

\blue{\begin{remk}
As we will work out in the proof of Prop.~\ref{prop:dXS} (see in particular formula \eqref{eq:dX(S)}), the pure gauge 
transformation $\del_X^A$ is not a symmetry of the considered model (given by $S=S_\YM+S_\gf+S_\gh$), in the sense that 
\be\label{eq:no-symmetry}
\del_X^A S \word{is {\bf not} of the form} \int dx\,\, X_a(x)\,\del_\mu J^\mu_a(x)
\word{for some currents $J_a\equiv (J_a^\mu)\in\sP^{\x 4}$.}
\ee
\end{remk}}

\medskip

\blue{However, for the ``extended local gauge transformation'' \eqref{eq:del_X} it indeed holds that
\be\label{eq:dXS}
\del_X S= \int dx\,\, X_a(x)\,\del_\mu J^\mu_a(x)\word{for some $J_a\in\sP^{\x 4}$,} 
\ee
which we call the ``non-Abelian gauge current(s)''. To derive this result}, we rewrite $\del_X F$ ($F\in\sF_\loc$) as
\begin{align}\label{eq:dXF}
\del_X F=&\int dx\,\,\tr\bigl(X(x)\cdot\sD(x)\bigr)F\word{where}\sD(x)=t_a \,\sD_a(x)\word{and}\\
\sD_a(x)F:=&\,\del^\mu_x\frac{\dl F}{\dl A_a^\mu (x)}-\ka \Bigl[A^\mu(x),\frac{\dl F}{\dl A^\mu (x)}\Bigr]_a
-\ka \Bigl[u(x),\frac{\dl F}{\dl u(x)}\Bigr]_a-\ka \Bigl[\tilde u(x),\frac{\dl F}{\tilde u(x)}\Bigr]_a\ ,\nonumber
\end{align}
where we use the notation
\be
\Bigl[\vf(x),\frac{\dl F}{\dl\vf(x)}\Bigr]_a:=f_{acb}\,\vf_c(x)\,\frac{\dl F}{\dl\vf_b(x)}=[t_a,\vf(x)]_b\,\frac{\dl F}{\dl \vf_b(x)}
\word{for $\vf=A^\mu, u, \tilde u$.}
\ee
\blue{With this preparation, the claim \eqref{eq:dXS} follows from the assumption \eqref{eq:g=1} and the following result:
\begin{prop}\label{prop:dXS}
For $x\in g^{-1}(1)^\circ\equiv\{z\in\bM\,\big\vert\,g(z)=1\}^\circ$  (where the upper index ``$\,\circ$'' denotes the 
interior of the pertinent set) it holds that
\begin{align}
&\sD_a(x)\,S_\YM=0\word{and}\label{eq:D(SYM)}\\
&\sD_a(x)\,S=\del_\mu J_a^\mu(x) \word{with the ``non-Abelian gauge current''}\label{eq:D(S_0+Sint)=dJ}\\
&J^\mu\equiv J^\mu_a t_a:=\la\square A^\mu+\ka\bigl(\la [(\del A),A^\mu]+[\tilde u,D^\mu u]\bigr)\ .\label{eq:J}
 \end{align}
 The first two terms of $J^\mu$ are the contribution coming from $\sD_a(x) S_\gf$ -- they depend on the gauge fixing parameter 
 $\la$, the last term of $J^\mu$ is the contribution of $\sD_a(x) S_\gh$. 
\end{prop}}

\begin{proof}
\blue{The first identity \eqref{eq:D(SYM)} follows directly from
\be
\int dx\,\,X_a(x)\,\sD_a(x)\,S_\YM=\del_X\,S_\YM=\int dy\,\,\del_X\,L_\YM(y)\overset{\eqref{eq:dX(LYM)}}{=}0\ ,
\quad\forall X_a\in\sD(\bM,\bR)\ .
\ee}
 
\blue{To verify the second relation \eqref{eq:D(S_0+Sint)=dJ}, we first compute the contribution coming from
$D^\mu(x)\,\frac{\dl }{\dl A^\mu(x)}$ to $\sD(x)(S_\gf+S_\gh)$: Since 
\be\label{eq:dlSgf}
\frac{\dl S_\gf}{\dl A_b^\mu(x)}=\la\,\del_\mu(\del A_b)(x)\ ,\quad
\frac{\dl S_\gh}{\dl A_b^\mu(x)}
=\ka[\del_\mu\tilde u(x),u(x)]_b\ ,
\ee 
we obtain
\begin{align}
D^\mu_{ab}(x)\,\frac{\dl S_\gf}{\dl A_b^\mu(x)}\,t_a=&\,\la\,\del_\mu\del^\mu (\del A)(x)+\la\,\ka[\del_\mu(\del A)(x),A^\mu(x)]
\nonumber\\
=&\,\del_\mu^x\bigl(\la\,\square A^\mu(x)+\la\,\ka[(\del A)(x),A^\mu(x)]\bigr)\ ,\label{eq:DSgf}\\
D^\mu_{ab}(x)\,\frac{\dl S_\gh}{\dl A_b^\mu(x)}\,t_a=&\ka\,\del^\mu_x [\del_\mu\tilde u(x),u(x)]
+\ka^2\,[[\del_\mu\tilde u(x),u(x)],A^\mu(x)]\ .\label{eq:DSgf}
\end{align}
Using additionally the result \eqref{eq:D(SYM)} and that $S_0+S_\inter = S_\YM+S_\gf+S_\gh$, we get
\begin{align}\label{eq:dX(S)}
&D^\mu_{ab}(x)\,\frac{\dl (S_0+S_\inter)}{\dl A_b^\mu(x)}\,t_a\\
&=\del_\mu^x \Bigl(\la\,\square A^\mu(x)+\ka\bigl(\la [(\del A)(x),A^\mu(x)]+[\del^\mu\tilde u(x),u(x)]\bigr)\Bigr)
+\ka^2\,[[\del_\mu\tilde u(x),u(x)],A^\mu(x)]\ .\nonumber
\end{align}
The fact that the $\ka^2$-term is not the divergence of a local field polynomial proves the statement \eqref{eq:no-symmetry}.}

\blue{To compute the contributions to $\sD(x)$ coming from $\del_X^{\tilde u}$ and $\del_X^u$, we use the results for
$\frac{\dl (S_0+S_\inter)}{\dl\tilde u}=\frac{\dl S_\gh}{\dl\tilde u}$ and 
$\frac{\dl (S_0+S_\inter)}{\dl u}=\frac{\dl S_\gh}{\dl u}$ given in \eqref{eq:FE-tilde-u} and
\eqref{eq:FE-u}, respectively; this yields
 \begin{align}
 -\ka& \Bigl[u(x),\frac{\dl (S_0+S_\inter)}{\dl u(x)}\Bigr]-\ka \Bigl[\tilde u(x),\frac{\dl (S_0+S_\inter)}{\tilde u(x)}\Bigr]\label{eq:dXu(S)}\\
 &=-\ka[\square \tilde u,u]-\ka^2[[\del_\mu\tilde u,A^\mu],u]+\ka[\tilde u,\square u]+\ka^2[\tilde u,\del_\mu[u,A^\mu]]
 \nonumber\\
 &=\ka\bigl(-\del_\mu[\del^\mu \tilde u,u]+\del_\mu[\tilde u,\del^\mu u]\bigr)+\ka^2\bigl(-[[\del_\mu\tilde u,A^\mu],u]
 +\del_\mu [\tilde u,[u,A^\mu]]-[\del_\mu\tilde u,[u,A^\mu]]\bigr)\ .\nonumber
 \end{align}
 Finally, in the sum \eqref{eq:dX(S)}$+$\eqref{eq:dXu(S)}, three of the four $\ka^2$-terms cancel out by the Jacobi identity for 
 the Lie bracket and we obtain the assertion \eqref{eq:D(S_0+Sint)=dJ}-\eqref{eq:J}.}
\end{proof}

\blue{The crucial relation \eqref{eq:D(S_0+Sint)=dJ} does not only state that $\del_X$ is a symmetry of the considered model,
but it also states that $J^\mu$  is the pertinent \emph{Noether current}, see \cite[Chap.~4.2.3]{D19} 
and Sect.~\ref{sec:global-trafo}. 
Obviously, $J^\mu$ is not an observable, since it contains the unphysical FP ghost fields and depends on the gauge fixing 
parameter $\la$. Note also that $(J_a^\mu)$ is a Lorentz- and a $\LieG$-vector (see \eqref{eq:LieG-vector} for the latter) and satisfies
\be\label{eq:ghost,dim(J)}
\dl_u(J^\mu)=0\ ,\quad (J^\mu)^*=J^\mu\word{and}\dim(J^\mu)=3\ ,
\ee
where $\dl_u(J^\mu)$ and $\dim(J^\mu)$ are the \emph{ghost number} and the \emph{mass dimension} of $J^\mu$, respectively. 
These notions are defined as follows: for a field \emph{monomial} $B\in\sP$, $\dl_u(B)$ is the number of factors 
$(\del^r)u$ minus the number of factors $(\del^s)\tilde u$ (where we ignore the partial derivatives $\del^r$ and $\del^s$);
and $\dim B$ is the number of basic fields (i.e., $A^\mu,u,\tilde u$) in $B$ plus the total number of partial derivatives 
on these basic fields, e.g., 
$\dim((\del_\mu A^\mu)A^\rho\cdot\del^r\tilde u\wedge\del^s u)=4+1+|r|+|s|$, see \cite[Def.~3.1.18]{D19}.
In \eqref{eq:ghost,dim(J)} we have extended these definitions in an obvious way to the field polynomial $J^\mu$, 
using that each monomial contributing to $J^\mu$ has the same ghost number and the same mass dimension.}

\medskip

We still need to prove that $X\mapsto \del_X$ is indeed a Lie algebra representation (which is crucial for the anomaly consistency 
condition of the MWI -- see Sect.~\ref{sec:ccAMWI}), i.e., we have to prove the following:

\begin{lema}\label{lem:dX=repr} For $X,Y\in\LieGc$ it holds that
\be\label{eq:[dX,dY]}
[\del_X,\del_Y]=\ka\,\del_{[X,Y]}\ ,
\ee
where we use $[\bullet,\bullet]$ for two different operations: on the l.h.s.~it is the commutator of functional differential operators
and on the r.h.s.~it is the Lie bracket.
\end{lema}

\begin{proof}
Obviously it holds that 
\be\label{eq:dXA+dXu+dX-u}
[\del_X,\del_Y]=[\del^A_X,\del^A_Y]+[\del_X^u,\del_Y^u]+[\del_X^{\tilde u},\del_Y^{\tilde u}]\ .
\ee
To compute the commutators on the r.h.s., we insert the definitions \eqref{eq:del_X}, this yields:
\begin{align}\label{eq:[dX,dY]=d[X,Y]}
[\del_X^A,\del_Y^A]=&\,\ka\int dx\,\,\Bigl([Y(x),(\del^\mu X(x)+
\ka [X(x),A^\mu(x)])]-(X\leftrightarrow Y)\Bigr)\,\frac{\dl}{\dl A^\mu (x)}\nonumber\\
=\,\ka\int dx&\,\,\Bigl(\del^\mu [Y(x),X(x)]+\ka\bigl([Y(x),[X(x),A^\mu(x)]]+[X(x),[A^\mu(x),Y(x)]]\bigr)\Bigr)\,\frac{\dl}{\dl A^\mu (x)}
\nonumber\\
=&\,\ka\,\del_{[X,Y]}^A\nonumber\\
[\del_X^u,\del_Y^u]=&\,\ka^2\int dx\,\,\Bigl([Y(x),[X(x),u(x)]]-(X\leftrightarrow Y)\Bigr)\,\frac{\dl}{\dl u(x)}\nonumber\\
=&\,\ka^2\int dx\,\,[[X(x),Y(x)],u(x)]\,\frac{\dl}{\dl u(x)}=\ka\,\del_{[X,Y]}^u
\end{align}
and similarly $[\del_X^{\tilde u},\del_Y^{\tilde u}]=\ka\,\del_{[X,Y]}^{\tilde u}$,
by using the Jacobi identity for the Lie bracket in each computation. Inserting these results into 
\eqref{eq:dXA+dXu+dX-u}, we obtain the assertion.
\end{proof}
The validity of \eqref{eq:[dX,dY]} is the reason for the global minus sign in the definition of $\del_X$ \eqref{eq:del_X}.
 


\section{Anomalous Master Ward Identity}\label{sec:AMWI}

\paragraph{Interacting quantum fields as formal power series in the interaction.} In view of perturbation theory,
we split the total action $S=S_\YM+S_\gf+S_\gh$ into a free and an interacting part, the former consists of the terms being
quadratic in the basic fields:
\begin{align}\label{eq:S=S0+Sint}
&S=S_0+S_\inter\word{with} 
\nonumber\\
&S_0:=\int dx\,\Bigl(-\tfrac{1}4\,\tr\bigl( (\del^\mu A^\nu(x)-\del^\nu A^\mu(x))\cdot (\del_\mu A_\nu(x)-\del_\nu A_\mu(x))\bigr)
\nonumber\\
&\qquad\qquad\quad\,\,-\tfrac{\la}2\,\tr\bigl(\del A)(x)\cdot (\del A)(x)\bigr)\Bigr)\ ,\\
&S_\inter\equiv S_\inter(g) :=\int dx\,\,\Bigl(\ka g(x)\,L_1(x)+(\ka g(x))^2\,L_2(x)\Bigr)\ ,\word{where}\nonumber\\
&L_1(x):=\tfrac{1}2\,\tr\bigl((\del^\mu A^\nu(x)-\del^\nu A^\mu(x))\cdot [A_\mu(x),A_\nu(x)]\bigr)+
\tr\bigl(\del_\mu \tilde u(x)\cdot [u(x),A^\mu(x)]\bigr)\ ,\nonumber\\
&L_2(x):=-\tfrac{1}4\,\tr\bigl([A^\mu(x),A^\nu(x)]\cdot [A_\mu(x),A_\nu(x)]\bigr)\ .\nonumber
\end{align}

We will study the anomalous MWI (AMWI) in terms of \emph{interacting fields}. The interacting field $G_F$ to 
the interaction $F\in\sF_\loc$ and belonging to $G\in\sF_\loc$, more precisely $G_F\vert_{F=0}=G$, is a
\emph{formal power series} in $F$, given in terms of the \emph{retarded product} $R\equiv (R_{n,1})$:%
\footnote{To simplify the notations, we write $\sF$ also for the space of formal power series with coefficients in $\sF$, 
and similarly for $\sF_\loc$.}
\be\label{eq:int-field}
G_F=R\bigl(e_\ox^{F /\hbar},G\bigr)\equiv\sum_{n=0}^\infty\frac{1}{n!\,\hbar^n}\,R_{n,1}(F^{\ox n},G)\in\sF
\ee
The retarded product is defined by the following axioms (for details see \cite{DF04} or \cite[Chap.~3.1]{D19}),
which we split into the basic axioms and renormalization conditions. The former are:
\footnote{Retarded and time ordered products are connected by Bogoliubov's formula, 
\be\label{eq:Bogoliubov}
R\bigl(e_\ox^{F/\hbar}\,,\,G\bigr)=\ovl T\bigl(e_\ox^{-i\,F/\hbar}\bigr)\star\  T\bigl(e_\ox^{i\,F/\hbar}\ox G\bigr)\ ,
\ee
where $\ovl T\bigl(e_\ox^{-i\,F/\hbar}\bigr)$ is the inverse w.r.t.~$\star$-product of the $S$-matrix
$\nS(F):=T\bigl(e_\ox^{i\,F/\hbar}\bigr)$. The basic axioms and renormalization
conditions, given here for the $R$-product, are equivalent to an off-shell version of Epstein-Glaser's axiomatic definition \cite{EG73} 
of the time ordered product, see \cite[Chap.~3.3]{D19}.}
\begin{itemize}
\item Linearity: $R_{n,1}: \sF_\loc^{n+1}\to\sF$ is linear.\\
\blue{The inductive construction of the sequence $(R_{n,1})_{n\in\bN}$ proceeds in terms of the $\sF$-valued distributions
$R_{n,1}\bigl(P_1(x_1)\ox\cdots,P(x)\bigr)\in\sD'(\bM^{n+1},\sF)$, $P_1,\dots P\in\sP$, which are related to 
$R_{n,1}: \sF_\loc^{n+1}\to\sF$ by
\be
\int dx_1\cdots dx\,\,R_{n,1}\bigl(P_1(x_1)\ox\cdots,P(x)\bigr)\,g(x_1)\cdots g(x)=R_{n,1}\bigl(P_1(g_1)\ox\cdots,P(g)\bigr)\ ,
\ee
where $P(g):=\int dx\,\,P(x)\,g(x)$, $g\in\sD(\bM)$. 
The requirement that $R_{n,1}$ depends only on (local) functionals, implies the Action Ward Identity (AWI):
\be\label{eq:AWI}
\del_{x_j}R(\cdots \ox P_j(x_j)\ox\cdots)=R(\cdots \ox(\del P_j)(x_j)\ox\cdots)\ ,\quad\forall P_j\in\sP\ ,
\ee
and similarly for the last entry $P(x)$; the AWI plays the role of an additional renormalization condition.}

\item Symmetry: $R_{n,1}$ is symmetrical in the first $n$ arguments,

\item Initial condition: $R_{0,1}(F)=F\ $,

\item Causality: $R(e_\ox^{(F+H) /\hbar},G)=R(e_\ox^{F /\hbar},G)\ $ if $\ (\supp G+\ovl{V}_-)\cap\supp H=\emptyset$,

\item GLZ-relation:
$$
\tfrac{1}{i}
\bigl[R\bigl(e_\ox^{G/\hbar}, F\bigr),
R\bigl(e_\ox^{G/\hbar}, H\bigr)\bigr]_\star
 = R\bigl(e_\ox^{G/\hbar} \ox F, H\bigr)-
R\bigl(e_\ox^{G/\hbar} \ox H, F\bigr) 
$$
(where on the l.h.s.~there is the commutator w.r.t.~the star product);
\end{itemize}
and the renormalization conditions read:
\begin{itemize}
\item Field independence: 
$\fd{}{x}R(e_\ox^{F /\hbar},G)=\frac{1}{\hbar}\,R(e_\ox^{F /\hbar}\ox\fd{F}{x},G)+R(e_\ox^{F /\hbar},\fd{G}{x})$
for $\vf=A^\mu_a,u_a,\tilde u_a$,

\item $*$- structure: $R(e_\ox^{F /\hbar},G)^*=R(e_\ox^{F^* /\hbar},G^*)$ ,

\item Poincar\'e covariance: $\bt_{\La,a}\circ R_{n,1}=R_{n,1}\circ\bt_{\La,a}^{\,\,\ox(n+1)}$ for the natural linear action
$\sP_+^\uparrow\ni(\La,a)\mapsto\bt_{\La,a}$ of $\sP_+^\uparrow$ on $\sF$,

\item $\LieG$-covariance: for $s_a:\sF\to\sF$ being the $\LieG$-rotation defined in \eqref{eq:LieG-rot}-\eqref{eq:LieG-rot-1},
it is required that
\be\label{eq:[sa,R]=0}
s_a\circ R_{n,1}=R_{n,1}\circ\sum_{k=1}^{n+1}(\id\oxyox s_a\oxyox\id)\quad\forall a\ ,
\ee
where (on the r.h.s.)~$s_a$ is the $k$th factor. In Prop.~\ref{pr:[Sa,R]=0} it is proved that this renormalization condition can be 
fulfilled by a symmetrization of $R_{n,1}$, which maintains the validity of all other renormalization conditions.

\item Preservation of ghost number: 
\blue{By $\dl_u$ we do not only denote the map giving the ghost number of a field monomial (as defined after \eqref{eq:ghost,dim(J)}),}
we also define $\dl_u$ as an operator $\dl_u:\sF\to\sF$ by $\dl_u F:=\int dx\,\Bigl(u_a(x)\wedge\frac{\dl F}{\dl u_a(x)}
-\tilde u_a(x)\wedge\frac{\dl F}{\dl\tilde u_a(x)}\Bigr)$. With this, it is required that
$$
\dl_u\Bigl(R_{n,1}\bigl(B_1(x_1)\ox\dots,B_{n+1}(x_{n+1})\bigr)\Bigr)=\Bigl(\sum_{j=1}^{n+1}\dl_u(B_j)\Bigr)\cdot 
R_{n,1}\bigl(B_1(x_1)\ox\dots,B_{n+1}(x_{n+1})\bigr)
$$
for all monomials $B_j\in\sP$.

\item Off-shell field equation: 
\be\label{eq:FE}
\square_\la \vf_{F}(x)= \square_\la\vf(x) + \bigl( \fd{F}{x} \bigr)_{F}\ ,
\ee 
where $\square_\la:=g^{\mu\nu}\square-(1-\la)\del^\mu\del^\nu$ if $\vf=A_\nu$
and $\square_\la:=\square$ if $\vf=u$ or $\vf=\tilde u$.

\item Almost homogeneous scaling: For all field {\it monomials} $B_1,\dots,B_{n+1}\in \sP$, the vacuum expectation values
$$
\sD'(\bR^{4n})\ni r(B_1,\dots;B_{n+1})(x_1 - x_{n+1},\dots):= \om_0\Bigl( R_{n,1}\bigl(
B_1(x_1) \ox\dots ;B_{n+1}(x_{n+1}) \bigr) \Bigr)
$$
scale almost homogeneously, i.e., homogeneously up to logarithmic terms (see e.g.~\cite[Def.~3.1.17]{D19}), 
with degree $\sum_{j=1}^n\dim B_j$, \blue{where the mass dimension of a field monomial is defined after \eqref{eq:ghost,dim(J)}.} 

\item $\hbar$-dependence: renormalization is done in each order of $\hbar$ individually.
\end{itemize}

\blue{By a direct inductive construction of the sequence $(R_{n,1})_{n\in\bN}$ (which is a version of the famous Epstein--Glaser 
construction \cite{EG73}), it has been proved that solutions of all these axioms exist (including the AWI), 
see \cite{DF04} or \cite[Chap.~3.2]{D19}. In this construction renormalization appears as the mathematical problem of
extending the distributions $r(x_1-x,\ldots):=\om_0\bigl(R_{n,1}(P_1(x_1)\ox\cdots,P(x))\bigr)$ from $\sD'(\bR^{4n}\setminus\{0\})$
to $\sD'(\bR^{4n})$ such that the renormalization conditions are maintained.
In addition, this construction yields \emph{all} solutions of the axioms. By an improved version of Stora's Main Theorem of perturbative
 renormalization (see \cite[Thm.~4.2]{DF04} or \cite[Thm.~3.6.3]{D19}) the set of solutions 
 can be understood as the orbit of the Stückelberg‒Petermann renormalization group \cite[Def.~3.6.1]{D19}
 when acting on a particular solution (any solution may be chosen as starting point).} 
 
\blue{For $F,G\sim\hbar^0$, the interacting field $G_F$ \eqref{eq:int-field} is a formal power series in 
 $\hbar$ -- this can be understood by studying the (Feyman) diagrams contributing to $R_{n,1}\bigl((F/\hbar)^{\ox n},G\bigr)$ \cite{DF01}:
 each interaction vertex $F/\hbar$ contains a factor $\hbar^{-1}$, but each propagator (i.e., each inner line) is accompanied by a
 factor $\hbar$ (cf.~\eqref{eq:star-product}); in addition, in contrast to the time ordered product,
 solely \emph{connected diagrams} contribute (due to the causal support of 
 $R_{n,1}$, see the axiom Causality) and  renormalization is done in each order of $\hbar$ individually.}

\paragraph{Anomalous Master Ward Identity (AMWI) for the extended, local gauge transformation $\del_X$.} 
We now apply the crucial theorem about the AMWI 
(see \cite[Sect.~5.2]{Brennecke08}, \cite[Thm.~5.2]{BrenneckeD09} or \cite[Chap.~4.3]{D19})
to the action $S=S_0+F$ (with $F\in\sF_\loc$ arbitrary) and the extended, infinitesimal, local gauge 
transformation $\del_X$. This yields the following results:
\footnote{Since the cited references use different notations for the anomaly map, we give the following identifications:
\begin{align}
q(x)Q(x)\equiv\vec q(x)\cdot\vec Q(x)=&\,-\bigl(D^\mu(x)X(x),\ka[X(x),u(x)],\ka[X(x),\tilde u(x)]\bigr)\\
\blue{-\int dx\,\,q(x)\,\Dl\bigl(e_\ox^F;Q(x)\bigr)\equiv}   -\Dl\bigl(e_\ox^F;qQ\bigr)=&\,\Dl X(F)\ ,
\end{align}
where $F\in\sF_\loc$ arbitrary and $\vec q\equiv (q_j)$ with $q_j\in\sD(\bM,\bR)$,  $\vec Q\equiv (Q_j)$ with 
$Q_j(x)\in\{\sum_aB_a(x)\,t_a\,\vert\,B_a\in\sP\}$. \blue{The minus sign in front of $\Dl X$ in the definition \eqref{eq:AMWI} 
of $\Dl X$ is chosen in accordance with references \cite{BDFR21,BDFR23}.}
Since $\Dl\bigl(e_\ox^F;qQ\bigr)$ depends linearly on $qQ$ we obtain
\begin{align}
\Dl X(F)=&\,\Dl\bigl(e_\ox^F;(\del^\mu X,0,0)\bigr)+\ka\Dl\bigl(e_\ox^F;([X,A^\mu],[X,u],[X,\tilde u])\bigr)\nonumber\\
=&\,\ka\int dx\,\,X_a(x)\,\Dl\bigl(e_\ox^F;([t_a,A^\mu(x)],[t_a,u(x)],[t_a,\tilde u(x)])\bigr)\ ,\label{eq:Dl=intdx}
\end{align}
where we use that the field shift does not contribute to the anomaly, i.e., 
$\Dl\bigl(e_\ox^F;(\del^\mu X,0,0)\bigr)$ for all $X\in\LieGc$, \blue{since generally it holds that
$\Dl\bigl(e_\ox^F;q\,1\bigr)=0$ (i.e., $Q=1$) due to the validity of the off-shell field equation for
the $R$-product (see \cite[Lemma 8]{Brennecke08} or the last part of \cite[Thm.~4.3.1(b)]{D19}).} 
The last expression in \eqref{eq:Dl=intdx} may be identified with the l.h.s.~of \eqref{eq:Dla(x)}.\label{fn:identify}} 
\begin{thm}\label{thm:AMWI}
\begin{itemize}
\item[(a)] Existence and uniqueness of the anomaly map. \blue{Given a retarded product $R$ satisfying the above listed basic 
axioms and renormalization conditions}, there exists a \emph{unique} sequence of \emph{linear} maps
$\Dl\equiv (\Dl^n)_{n\in\bN}$,
\be\label{eq:anom-terms}
\Dl^n :\sF_\loc^{\ox n}\ox\LieGc\longrightarrow \sF_\loc
\ee
which are invariant under permutations of the first $n$ factors and fulfill the AMWI:
\begin{align}\label{eq:AMWI}
R\Bigl(e_\ox^{F /\hbar},&\,\bigl(\del_X(S_0+F)-\Dl X(F)\bigr)\Bigr)=\int dx\,\,
R\Bigl(e_\ox^{F /\hbar},\,\del_X A_a^\mu(x)\Bigr)\cdot\frac{\dl S_0}{\dl A_a^\mu(x)}\nonumber\\
&+R\Bigl(e_\ox^{F /\hbar},\,\del_X u_a(x)\Bigr)\cdot\frac{\dl S_0}{\dl u_a(x)}
+R\Bigl(e_\ox^{F /\hbar},\,\del_X \tilde u_a(x)\Bigr)\cdot\frac{\dl S_0}{\dl \tilde u_a(x)}\ ,
\end{align}
where $\Dl X(F)$ is the formal power series
\be
\Dl X(F):=\sum_{n=1}^\infty \frac{1}{n!}\, \Dl^n(F^{\ox n}; X)\ .
\ee
Obviously, the ``anomaly map'' $\Dl X:\sF_\loc\to\sF_\loc$%
\footnote{Sometimes we use the name ``anomaly map(s)'' also for the maps $\Dl^n$ \eqref{eq:anom-terms} and also for
$\Dl:\sF_\loc\ox\LieGc\to\sF_\loc\,;\,(F;X)\mapsto\Dl X(F)$.}
 depends on the renormalization prescription for the retarded product $R$. 
For its derivative we write:
\be\label{eq:DlX'}
\langle(\Dl X)'(F),G\rangle:=\frac{d}{d\la}\Big\vert_{\la =0}\Dl X(F+\la G)=
\sum_{n=0}^\infty \frac{1}{n!}\, \Dl^{n+1}(F^{\ox n}\ox G; X)
\ee
for all $F,G\in\sF_\loc$.

\item[(b)] Properties of the anomaly map. \blue{The so-defined maps $\Dl^n$ have the following propwerties:}
\begin{itemize}

\item[(b1)] (Locality and Translation covariance)
There exist \emph{linear} maps $P^n_{a,r}\:\sP^{\ox n}\to \sP$  
which are \emph{uniquely} determined by
\begin{align}\label{eq:Pnar}
\Dl^n\bigl(\ox_{j=1}^n & B_j(g_j); X\bigr) =\int dx\,\,X_a(x)\int dy_1\cdots dy_n\,\,g_1(y_1)\cdots g_n(y_n)\nonumber\\
&\cdot\sum_{r\in (\bN_0^4)^n}\del^r\dl(y_1-x,\ldots,y_n-x)\,P^n_{a,r}(\ox_{j=1}^n B_j)(x)\ ,
\end{align}
where $B(g):=\int dy\,B(y)\,g(y)\in\sF_\loc$ (with $B,g$ as in \eqref{eq:local-fields}), $X(x):=X_a(x)\,t_a$ and
the sum over ``$r$'' is finite. In addition, in the particular case that $B_k=c\in\bR$ for some $1\leq k\leq n$, it holds that
$P^n_{a,r}(\ox_{j=1}^n B_j)=0$ for all $a,r$ and $n$.

\item[(b2)] ($\hbar$-dependence) $\Dl X(F)=\sO(\hbar)$ if $F\sim\hbar^0$.

\item[(b3)] (Field independence) 
\be\label{eq:FI-anomaly}
\frac{\dl\Dl X(F)}{\dl \vf(y)}=\Big\langle(\Dl X)'(F),\frac{\dl F}{\dl \vf(y)}\Big\rangle\word{for}\vf=A^\mu_a,u_a,\tilde u_a\ .
\ee

\item[(b4)] (Homogeneous scaling) If in \eqref{eq:Pnar} all $B_j$ are homogeneous in the mass dimension, then
$P^n_{a,r}(\ox_{j=1}^n B_j)$ is also homogeneous in the mass dimension and the multi-derivative $r$ is of order
\be\label{eq:Dl-homogen}
|r|=\sum_{j=1}^n\dim (B_j)-4(n-1)-\dim \bigl(P^n_{a,r}(\ox_{j=1}^n B_j)\bigr)\ .
\ee

\item[(b5)] (Lorentz covariance) In terms of $\Dl_a(x)$ (defined below in \eqref{eq:Dla(x)}) this reads:
\be
\bt_\La\bigl(\Dl_a(x)(F)\bigr)=\Dl_a(\La x)\bigl(\bt_\La(F)\bigr)\word{for all} \La\in\sL^\uparrow_+,\,F\in\sF_\loc\ .
\ee

\item[(b6)] (Consequence of Off-shell field equation) As a consequence of the validity of the Off-shell field equation for the $R$-product,
it holds that
\footnote{For $B\in\sP$, $\langle(\Dl X)'(F),B(y)\rangle$ is defined in the obvious way (see \eqref{eq:DlX'}-\eqref{eq:Pnar}): 
\be\label{eq:DlX'-B}
\int dy\,\,g(y)\,\langle(\Dl X)'(F),B(y)\rangle:=\langle(\Dl X)'(F),B(g)\rangle\word{for all $g$ as in \eqref{eq:local-fields}.}
\ee}
\be\label{eq:Dl(vfox...)} 
\langle(\Dl X)'(F),\del^r\vf(x)\rangle=0\word{for} r\in\bN_0^4,\,\,\vf=A^\mu,\,u,\,\tilde u\ .
\ee
\end{itemize}
\end{itemize}
\end{thm}

\blue{The last statement in \emph{(b1)} can equivalently be written as $\Dl^n(F^{\ox (n-1)}\ox c;X)=0$ for $c\in\bC$. 
The latter follows from locality of $\Dl X$, see e.g.~\cite[Thm.~7,~property~(ii)]{Brennecke08} 
and take into account that $\fd{c}{x}=0$.
A proof of \emph{(b6)} is not yet published anywhere (only a weaker version is verified in \cite[Exer~4.3.3]{D19}), 
we close this gap in the Appendix. We explain there also, why the Action Ward Identity (AWI) \eqref{eq:Dl-AWI}
holds for the anomaly maps $\Dl^n(\bullet ;X)$.} 

\medskip
As we see from \eqref{eq:Pnar}, we may define a map $\Dl_a(x):\sF_\loc\to\sD'(\bM,\sF_\loc)$ by
\be\label{eq:Dla(x)}
\sum_a\int dx\,\,X_a(x)\,\Dl_a(x)(F):=\Dl X(F)\ , \quad\forall F\in\sF_\loc\ ,\quad X(x)=X_a(x)\,t_a\in\LieGc\ ;
\ee
obviously, its derivative,
\be\label{eq:Dl'}
\langle \Dl_a(x)'(F),\,G\rangle:=\frac{d}{d\la}\Big\vert_{\la =0}\Dl_a(x)(F+\la G)\ , \quad\forall F,G\in\sF_\loc\ ,
\ee
satisfies the analogous relation $\sum_a\int dx\,\,X_a(x)\,   \langle \Dl_a(x)'(F),\,G\rangle=\langle(\Dl X)'(F),G\rangle$.

Taking into account also the relations \eqref{eq:del_X-D} and \eqref{eq:dXF}, we may write the AMWI 
\eqref{eq:AMWI} as
\begin {align}\label{eq:AMWI-x}
R\Bigl(e_\ox^{F /\hbar},&\,\Bigl(\sD_a(x)\,(S_0+F)-\Dl_a(x)(F)\Bigr)\Bigr)\\
=&\,R\Bigl(e_\ox^{F /\hbar},\, D^\mu_{ab}(x)\Bigr)\cdot\Bigl(\square A_{\mu,b}(x)-(1-\la)\del_\mu(\del A_b)(x)\Bigr)\nonumber\\
&+\ka\,f_{abc}\, R\Bigl(e_\ox^{F /\hbar},\,u_c(x)\Bigr)\cdot\square\tilde u_b(x)
-\ka\, f_{abc}\,R\Bigl(e_\ox^{F /\hbar},\,\tilde u_c(x)\Bigr)\cdot\square u_b(x)\nonumber
\end{align}
for $x\in g^{-1}(1)^\circ$, with
\be
R\Bigl(e_\ox^{F /\hbar},\,D^\mu_{ab}(x)\Bigr):=\dl_{ab}\del_x^\mu+
\ka\,f_{abc}\,R\Bigl(e_\ox^{F /\hbar},\,A^\mu_c(x)\Bigr)\ .
\ee

\medskip

\paragraph{Restriction to the Yang-Mills interaction with FP ghosts.}
Of particular interest is the case $F=S_\inter$. \blue{As proved in Prop.~\ref{prop:dXS}, $\del_X$ is a classical symmetry of this model 
and the pertinent Noether current is the non-Abelian gauge current $J^\mu$ \eqref{eq:J}; therefore,
the AMWI \eqref{eq:AMWI-x} for $F=S_\inter$ can be interpreted as conservation of the corresponding interacting, quantum gauge 
current (up to anomalies) -- this is essentially the content of the following Corollary.}

\begin{corl}\label{prop:AMWI}
 For $x\in g^{-1}(1)^\circ$ the AMWI for the extended local gauge transformation $\del_X$ \eqref{eq:del_X} and the 
 Yang-Mills-interaction with ghosts, $S_\inter$, can be written as
\begin{align}\label{eq:AMWI-explicit}
&\del_\mu^x R\Bigl(e_\ox^{S_\inter /\hbar},\, J_a^\mu(x)\Bigr)=R\Bigl(e_\ox^{S_\inter /\hbar},\,\Dl_a(x)(S_\inter)\Bigr)\\
&\qquad+R\Bigl(e_\ox^{S_\inter /\hbar},\, D^\mu_{ab}(x)\Bigr)\cdot\Bigl(\square A_{\mu,b}(x)
-(1-\la)\del_\mu(\del A_b)(x)\Bigr)\nonumber\\
&\qquad +\ka\,f_{abc}\, R\Bigl(e_\ox^{S_\inter /\hbar},\,u_c(x)\Bigr)\cdot\square\tilde u_b(x)
-\ka\, f_{abc}\,R\Bigl(e_\ox^{S_\inter /\hbar},\,\tilde u_c(x)\Bigr)\cdot\square u_b(x)\ .\nonumber
\end{align}
Working on-shell (i.e., restricting this identity to configurations solving the free field equations), the last three terms on the 
r.h.s.~vanish; hence, this identity states that the interacting gauge current is conserved up to anomalies (given by the 
first term on the r.h.s.):
\be
\del_\mu^x\,J^\mu_{a,S_\inter,0}(x)=\bigl(\Delta_a(x)(S_\inter)\bigr)_{S_\inter,0}
 \word{for $x\in g^{-1}(1)^\circ$,}
 \ee
 where \blue{the lower index `$0$' denotes restriction of the functionals to $\sC_0$, see} \eqref{eq:on-shell}.
\end{corl}

\begin{proof}
The claim \eqref{eq:AMWI-explicit} follows immediately from Prop.~\ref{prop:dXS} and the version \eqref{eq:AMWI-x}
of the AMWI; to wit, by inserting the crucial relation \eqref{eq:D(S_0+Sint)=dJ} into the latter.
\end{proof}

A contribution to $\Dl_a(x)(S_\inter)$ being of the form 
\be\label{eq:triv-anom}
\del^x_\mu \tilde P^\mu_a(x)\word{with} \tilde P^\mu_{a}=\sum_{l=1}^\infty\ka^l \,\tilde P_{a,l}^\mu,\quad \tilde P^\mu_{a,l}\in\sP
\word{is a ``trivial'' anomaly,}
\ee
if  $\tilde P^\mu_{a}$ is Lorentz- and $\LieG$-covariant and fulfills $\dim \tilde P^\mu_a=3$,  $\dl_u(\tilde P_a^\mu)=0$
and $(\tilde P_a^\mu)^*=\tilde P_a^\mu$;
since it can be removed by adding terms $\sim\ka^l$ with $l\geq 1$ to $J^\mu_a$, i.e.,
\be\label{eq:J-renormal}
J^\mu_a\to \tilde J^\mu_a:=J^\mu_a-\tilde P^\mu_a\ .
\ee
\blue{Most of} these assumptions about $\tilde P^\mu_{a}$ are no restriction. 
To wit, below (in \eqref{eq:form-Dl(Sint)} and subsequently) we explain 
that $P_a(x):=\Dl_a(x)(S_\inter)\in\ka\sP\pw{\ka}$ fulfills $\dim(P_a)=4$, $\dl_u(P_a)=0$,
$(P_a)^*=P_a$ and is Lorentz-invariant and a $\LieG$-vector (see \eqref{eq:LieG-vector}). \blue{That is, $\del^x_\mu \tilde P^\mu_a$
has these properties and we may choose $\tilde P^\mu_a$ such that $\dim \tilde P^\mu_a=3$,  $\dl_u(\tilde P_a^\mu)=0$
and $(\tilde P_a^\mu)^*=\tilde P_a^\mu$. (For example, the latter property can be achieved by replacing $\tilde P_a^\mu$
by $\tfrac12(\tilde P_a^\mu+(\tilde P_a^\mu)^*$, using that $\del_\mu(\tilde P_a^\mu)^*=\del_\mu\tilde P_a^\mu$.)
In additon, working always manifestly Lorentz- and $\LieG$-covariant, we may expect 
that the arising $\tilde P_a^\mu$ are also Lorentz- and $\LieG$-covariant.}

\medskip

The finite renormalization of $J^\mu_a$ \eqref{eq:J-renormal} can also be interpreted as a \emph{finite renormalization of the $R$-product:}
$R\mapsto \wh R$. To wit, let $Z$ be an element of the Stückelberg-Petermann renormalization group $\sR$ (see
\cite[Def.~3.6.1]{D19}) satisfying
\footnote{\blue{A proof that there exists a $Z$ fufilling all the defining properties of $\sR$ and \eqref{eq:Z} is not worked out here -- 
we do not see any serious obstacle. For example, Field Independece of such a $Z$ is satisfied by constructing $Z$ in terms of the
(causal) Wick expansion.}}
\be\label{eq:Z}
Z(S_\inter)=S_\inter\word{and} \langle Z'(S_\inter),J_a^\mu(x)\rangle = J^\mu_a(x)-\tilde P^\mu_a\ .
\ee
The second part of the Main Theorem of perturbative renormalization \cite[Thm.~3.6.3, Exer.~3.6.11]{D19} 
states that (for any $Z\in\sR$) the linear map $\wh R:\ovl{\bC\oplus\bigoplus_{n=1}^\infty\sF_\loc^{\ox n}}\ox\sF_\loc\to\sF$, 
which is symmetrical in all arguments except the last one, and is defined by 
\be
\wh R\bigl(e_\ox^{F/\hbar},G\bigr):=R\bigl(e_\ox^{Z(F)/\hbar},\langle Z'(F),G\rangle\bigr)\ ,
\ee 
satisfies also the axioms for an $R$-product. For our particular $Z\in\sR$ \eqref{eq:Z} we obtain
\be\label{eq:hat(R)}
\wh R\bigl(e_\ox^{S_\inter/\hbar},J^\mu_a(x)\bigr):=R\bigl(e_\ox^{Z(S_\inter)/\hbar},\langle Z'(S_\inter),J_a^\mu(x)\rangle\bigr)
=R\bigl(e_\ox^{S_\inter/\hbar},(J_a^\mu-\tilde P_a^\mu)(x)\bigr).
\ee


\begin{remk}
As shown in  \cite[Prop.~13]{Brennecke08}, one can reach that 
\be
\langle\Dl_a(x)'(0),S_\inter\rangle=0
\ee 
by a finite admissible renormalization of $R_{1,1}$ \blue{(remember \eqref{eq:int-field})}. With that, we have $\Dl_a(x)(S_\inter)=\sO(S_\inter^{\,\,\,2})=\sO(\ka^2)$, 
that is, the ``renormalization'' of $J$ \eqref{eq:J-renormal} does not affect the $\ka^1$-term of $J$.
\end{remk}


\begin{remk}\label{rem:class-dJ=0}
Since the MWI holds in classical field theory, the perturbative, classical, retarded field
\footnote{\blue{Compared to \eqref{eq:MWI-class}-\eqref{eq:prod-class-int-fields} we somewhat simplify the notation:
we write $B_{S_\inter}^\class(x)$ instead of $\bigl(B(x)\bigr)_{S_\inter}^\class(x)$ (for $B\in\sP$). 
In \cite{D19} the  perturbative, classical, retarded fields are denoted by $B_{S_\inter}^{\mathrm{ret}}(x)$; 
here we think that the notation $B_{S_\inter}^\class(x)$ is clearer.}}
\be\label{eq:class-int-field}
J^{\mu\,\ret}_{S_\inter,0}(x)\equiv J^{\mu\,\ret}_{a,S_\inter,0}(x)\,t_a :=
R\Bigl(e_\ox^{S_\inter /\hbar},\, J^\mu(x)\Bigr)_0\Big\vert_{\hbar=0}
\ee
(where the lower index `$0$' denotes restriction to $\sC_0$, see \eqref{eq:on-shell}) 
fulfills the on-shell version of \eqref{eq:AMWI-explicit} without anomaly, i.e., $\del_\mu^x J^{\mu\,\ret}_{S_\inter,0}(x)=0$
\blue{for $x\in g^{-1}(1)^\circ$}. It is instructive to verify this explicitly. To do this, we use the field equations 
\eqref{eq:FE-A}-\eqref{eq:FE-u} and the classical factorization 
\be\label{eq:factorization}
(B\cdot C)^\ret_{S_\inter,0}(x)=B^\ret_{S_\inter,0}(x)\cdot C^\ret_{S_\inter,0}(x)
\word{for} B,C\in\sP\ .
\ee
In detail, writing $A^\mu$ for $A^{\mu\,\ret}_{S_\inter,0}(x)$ and analogously for $u$ and $\tilde u$, we obtain%
\footnote{Note that by e.g. $[A_\mu,\square A^\mu](x)$ we mean
\begin{align*}
[A_{\mu,S_\inter,0}^\ret (x),\square_x A^{\mu,\ret}_{S_\inter,0}(x)]=&\,
t_a f_{abc}\,A_{b,\mu,S_\inter,0}^\ret (x)\cdot\square_x A^{\mu,\ret}_{c,S_\inter,0}(x)
\overset{\eqref{eq:factorization}}{=}t_a f_{abc}(A_{b,\mu}\,\square A_c^\mu)^\ret_{S_\inter,0}(x)\\
=&\,[A_\mu,\square A^\mu]^\ret_{S_\inter,0}(x)\ .
\end{align*}}
\begin{align}\label{eq:1}
\la\,\del_\mu\square\, A^\mu\overset{\eqref{eq:FE-A}}{=}&\ka\bigl(-[A_\nu,\square A^\nu]+[A_\nu,\del^\nu(\del A)]
-\del_\mu[\del^\mu\tilde u,u]\bigr)\nonumber\\
&-\ka^2\bigl([\del_\mu A_\nu,[A^\nu,A^\mu]]+[A_\nu,\del_\mu[A^\nu,A^\mu]]\bigr)\ .
\end{align}
Next we insert the field equation \eqref{eq:FE-A} into the first term on the r.h.s.:
\begin{align}\label{eq:2}
\ka\,[A_\mu,\square A^\mu]\overset{\eqref{eq:FE-A}}{=}&\,\ka(1-\la)[A_\mu,\del^\mu(\del A)]\nonumber\\
&+\ka^2\bigl([A_\mu,\del_\nu [A^\nu,A^\mu]]-[A_\mu,[A_\nu,(\del^\mu A^\nu-\del^\nu A^\mu)]]
-[A_\mu,[\del^\mu\tilde u,u]]\bigr)\nonumber\\
&+\ka^3[A_\mu,[A_\nu,[A^\nu,A^\mu]]]\ .
\end{align}
The $\ka^3$-term vanishes by the antisymmetry and Jacobi identity of the Lie bracket, the last term in \eqref{eq:1} 
and the first $\ka^2$-term in \eqref{eq:2} cancel out, the 
second last term of \eqref{eq:1} and the second $\ka^2$-term  of \eqref{eq:2} cancel also out by the 
antisymmetry and Jacobi identity of the Lie bracket. With this result for $\la\,\del_\mu\square\, A^\mu$, we obtain
\begin{align}\label{eq:dJ}
\del_\mu J^\mu=&\,\ka\Bigl( (\la-1)[A_\mu,\del^\mu(\del A)]+[A_\nu,\del^\nu(\del A)]-\del_\mu[\del^\mu\tilde u,u]
+\la\,\del_\mu[(\del A),A^\mu]+\del_\mu[\tilde u,\del^\mu u]\Bigr)\nonumber\\
&\,+\ka^2\Bigl([A_\mu,[\del^\mu\tilde u,u]]+\del_\mu[\tilde u,[u,A^\mu]]\Bigr)\ .
\end{align}
Obviously, the first, second and fourth $\ka$-term cancel out. By using the field equations \eqref{eq:FE-tilde-u} and \eqref{eq:FE-u}, the sum of the remaining $\ka$-terms gives
\begin{align}
\ka\bigl(-\del_\mu[\del^\mu\tilde u,u]+\del_\mu[\tilde u,\del^\mu u]\bigr)=&\,\ka\bigl(
-[\square \tilde u,u]+[\tilde u,\square u]\nonumber\\
=&\,\ka^2\bigl([[\del_\mu\tilde u,A^\mu],u]-[\tilde u,\del_\mu[u,A^\mu]]\bigr)\ .
\end{align}
Inserting this result into \eqref{eq:dJ} we end up with
\be
\del_\mu J^\mu=\ka^2\bigl([A_\mu,[\del^\mu\tilde u,u]]+[\del_\mu\tilde u,[u,A^\mu]]
+[[\del_\mu\tilde u,A^\mu],u]\bigr)=0
\ee
by means of the Jacobi identity for the Lie bracket.

Trying to transfer this derivation to pQFT, this is possible for the first step \eqref{eq:1}, but for all further steps
this does not work, since the classical factorization \eqref{eq:factorization} does not hold in pQFT.
Thus, the MWI \eqref{eq:AMWI-explicit} is a non-trivial Ward identity; even if the $R$-product 
is renormalized such that the off-shell field equations for $ A_{S_\inter}^\mu, \, u_{S_\inter},\,\tilde u_{S_\inter}$
hold, it may be violated by an anomaly term.
\end{remk}

\paragraph{The most general form of the anomaly $\Delta_a(x)(S_\inter)$.} 
Applying \eqref{eq:Pnar} to $\Delta_a(x)(S_\inter)$ and setting
$g_1(x):=g(x)$ (multiplying $L_1$) and $g_2(x):=(g(x))^2$ (multiplying $L_2$) we obtain
\begin{align}\label{eq:int-Dln}
\Delta_a(x)(S_\inter)=&\sum_{n=1}^\infty\frac{1}{n!}\sum_{j_1,\ldots,j_n=1,2}\ka^{j_1+\ldots+j_n}
\int dy_1\cdots dy_n\,\,g_{j_1}(y_1)\cdots g_{j_n}(y_n)\nonumber\\
&\cdot \Dl^n_a\bigl(\ox_{k=1}^n L_{j_k}(y_k);x\bigr)\word{where}\nonumber\\
\Dl^n_a\bigl(\ox_{k=1}^n L_{j_k}(y_k);x\bigr)=&\sum_r\del^r\dl(y_1-x,\ldots,y_n-x)\,P^n_{a,r}(L_{j_1}\oxyox L_{j_n})(x)\ ,
\end{align}
with $P^n_{a,r}(L_{j_1}\oxyox L_{j_n})\in\sP$; due to \eqref{eq:Dl-homogen} and $\dim(L_1)=4=\dim(L_2)$ it holds that
\be
\dim \bigl(P^n_{a,r}(L_{j_1}\oxyox L_{j_n})\bigr)=4-|r|\ .
\ee
Since $x\in g^{-1}(1)^\circ$ only $r=(0,\ldots,0)$ contributes
\footnote{\blue{Note that
\begin{align*}
\int & dy_1\cdots \,\,g_{j_1}(y_1)\cdots \del_{y_1}^\mu\dl(y_1-x,\ldots)\,P(x)\\
&=-\int dy_1\cdots \,\,\del^\mu g_{j_1}(y_1)\cdots \dl(y_1-x,\ldots)\,P(x)
=\del^\mu g_{j_1}(x)\cdots \,P(x)=0\ ,
\end{align*}
since $x\in g^{-1}(1)^\circ$.}}
 and it remains
\begin{itemize}
\item[(i)]\be\label{eq:form-Dl(Sint)}
\Delta_a(x)(S_\inter)=P_a(x)\word{with $P_a\in\ka\,\sP\pw{\ka}$ and $\dim P_a=4$.}
\ee 
\end{itemize}

\noindent 
In addition, we know that $P_a$ satisfies the following properties:
\begin{itemize}
\item[(ii)] Since $\Dl_a(x)(F)$ satisfies Field independence (Thm.~\ref{thm:AMWI}(b3)) it also satisfies the
causal Wick expansion w.r.t.~$F$, see \cite[Chap.~3.1.4]{D19}. For $\Delta_a(\bullet)(S_\inter)=P_a(\bullet)$ this implies that
$P_a$ is a polynomial in the basic fields appearing in $L_1$ and $L_2$ (i.e., $A^\mu,\,u$ and $\tilde u$)
and partial derivatives of these basic fields.
\item [(iii)] $P_a$ is Lorentz invariant, this follows immediately from Lorentz covariance of $\Dl X$ (Thm.~\ref{thm:AMWI}(b5)).
\item [(iv)] $P(x):=P_b(x)t_b$ is a $\LieG$-vector, see \eqref{eq:LieG-vector}. 
\item[(v)] $P_a$ is real: $P_a^*=P_a$.
\item[(vi)] The ghost number is $\dl_u(P_a)=0$.
\end{itemize}

The last four properties are obtained as follows: 
One proves a corresponding statement for $\Dl^n\bigl(\ox_{k=1}^n L_{j_k}(y_k);x\bigr):=\Dl^n_b(\cdots)t_b$ \eqref{eq:int-Dln}
by induction on $n$, using the inductive construction of the sequence of maps $(\Dl^n)$
in the proof of Thm.~\ref{thm:AMWI} (given in \cite[formula (5.15)]{Brennecke08} or \cite[formula (4.3.11)]{D19}), explicitly
\begin{align}\label{eq:recursion-Dln}
\Dl^n\bigl(&\ox_{k=1}^n L_{j_k}(y_k);x\bigr)
= \Bigl(R\Bigl(\bigox_{k=1}^n L_{j_k}(y_k),\, D^\mu_{bc}(x)t_b\Bigr)\cdot
\frac{\dl S_0}{\dl A^\mu_c(x)}+\ldots\Bigr)\nonumber\\
&-R\Bigl(\bigox_{k=1}^n L_{j_k}(y_k),\, \Bigl(D^\mu_{bc}(x)t_b
\cdot\frac{\dl S_0}{\dl A^\mu_c(x)}+\ldots\Bigr)\Bigr)\nonumber\\
&-\sum_{l=1}^n R\Bigl(\bigox_{k=1\,(k\neq l)}^n
L_{j_k}(y_k),\, \Bigl(D^\mu_{bc}(x)t_b\cdot\frac{\dl L_{j_l}(y_l)}{\dl A^\mu_c(x)}+\ldots\Bigr)\Bigr)\nonumber\\
&-\sum_{I\subset\{1,\dots,n\}\,,\,I^\comp \neq \emptyset}
R\Bigl(\bigox_{r\in I^\comp} L_{j_r}(y_r),\,
\Dl_b^{|I|}\bigl(\ox_{k\in I} L_{j_k}(y_k); x\bigr)t_b\Bigr)\ ,
\end{align}
where the dots stand for the two terms which are obtained by replacing $D^\mu_{bc}(x)t_b\cdot\frac{\dl }{\dl A^\mu_c(x)}$
by $\ka f_{bcd}u_d(x)t_b\cdot\frac{\dl}{\dl u_c(x)}$ and $\ka f_{bcd}\tilde u_d(x)t_b\cdot\frac{\dl}{\dl \tilde u_c(x)}$, respectively.
With that we see that the claims follow from the validity of the corresponding renormalization conditions for the $R$-product
and the corresponding properties of $L_1$ and $L_2$. In particular, to verify (iv) one proves that
\be
s_a\,\Dl^n\bigl(\ox_{k=1}^n L_{j_k}(y_k);x\bigr)=[t_a\,,\,\Dl^n\bigl(\ox_{k=1}^n L_{j_k}(y_k);x\bigr)]
\ee
by induction on $n$; besides \eqref{eq:recursion-Dln} the validity of \eqref{eq:[sa,R]=0}, $s_a\,L_j=0$ (for $j=1,2$) \blue{and 
the $\LieG$-covariance  of the further appearing expressions (i.e.,
$\frac{\dl S_0}{\dl A^\mu_c(x)}$, $\frac{\dl L_{j_l}(y_l)}{\dl A^\mu_c(x)}$, $\Dl_b^{|I|}(\cdots)$ are $\LieG$-vectors and 
$D^\mu_{bc}(x)$ is a $\LieG$-tensor of 2nd rank, see \eqref{eq:LieG-vector}-\eqref{eq:LieG-tensor})
are used for that. (Note also that $R(\ox_jF_j,\, F_bt_b)=R(\ox_jF_j,\, F_b)\,t_b$ for $\bC$-valued fields $F_j$ and $F_b$.)}

\medskip

As explained in \eqref{eq:triv-anom}-\eqref{eq:J-renormal}, contributions to $P_a$ 
being of the form $\del_\mu \tilde P^\mu_a$ can be neglected.
We give some examples for possible, nontrivial contributions to $P_a$, that is, field monomials satisfying (i)-(vi):
\begin{itemize}
\item[(a)] terms bilinear in the basic fields: $g^1_{abc}\, (\del^\mu A_b^\nu-\del^\nu A_b^\mu) 
(\del_\mu A_{c,\nu}-\del_\nu A_{c,\mu})$, 
$\,g^2_{abc}\, \del^\mu\tilde u_b \del_\mu u_c$,
\item[(b)] terms trilinear in the basic fields: $g^3_{abcd}\,A_{b,\mu} (\del^\mu A_c^\nu-\del^\nu A_c^\mu) A_{d,\nu}$,
$\,g^4_{abcd}\,A_{b,\mu} \del^\mu\tilde u_c  u_d$,
\item[(c)] terms quadrilinear in the basic fields: $g^5_{abcde}\, A^{\mu}_b A^{\nu}_c A_{d,\mu}A_{e,\nu}$,
$\,g^6_{abcde}\, A^{\mu}_b A_{c,\mu} \tilde u_d u_e$,
\end{itemize}
with $\LieG$-covariant tensors $g^k_{\cdots}\in\bR$ of rank $3,4$ or $5$, respectively.

 
\section{Consistency condition for the anomaly of the MWI}\label{sec:ccAMWI}

\paragraph{Anomaly consistency condition for the extended, local gauge transformation $\del_X$.}
Due to the knowledge about the anomaly  $\Delta_a(x)(S_\inter)=P_a(x)$ obtained so far, there still remain a lot of 
field monomials which may contribute to $P_a$. The main message of this paper is that, ignoring ``trivial'' contributions (i.e.,
terms of the form $\del_\mu\tilde P^\mu_a$), \emph{all possible contributions to $P_a$ are excluded by the anomaly
consistency condition} derived in \cite{BDFR23}.

Since $\del_X A^\mu(x)$,  $\del_X u(x)$ and $\del_X \tilde u(x)$
are at most linear in the basic fields (i.e., $\del_X$ is an affine field transformation)
the derivation of the anomaly consistency condition given in \cite[Sect.~4]{BDFR23} applies to the AMWI \eqref{eq:AMWI}
(the only modification is that the factor $\ka$ appearing in the relation \eqref{eq:[dX,dY]} yields a factor $\ka$
on the l.h.s.~of \eqref{eq:WZ}); hence the anomaly map $\Dl$ satisfies 
\begin{equation}
\begin{split}\label{eq:WZ}
    \ka\,\Delta([X,Y])(F)&=\big\langle(\Delta Y)'(F),\Delta X(F) \big\rangle-\big\langle(\Delta X)'(F),\Delta Y(F) \big\rangle\\ 
    &\quad+\partial_X(\Delta Y(F))-\partial_Y(\Delta X(F))\\
    &\quad-\big\langle(\Delta Y)'(F),\partial_X(S_0+F)\big\rangle +\big\langle(\Delta X)'(F),\partial_Y(S_0+F)\big\rangle\ .
\end{split}
\end{equation}
Proceeding analogously to the step from \eqref{eq:AMWI} to \eqref{eq:AMWI-x} and using 
the definition \eqref{eq:DlX'-B} for $\Dl_a(x)'(F)$ in place of $(\Dl X)'(F)$,
the anomaly consistency condition \eqref{eq:WZ} can be written as
\begin{align}\label{eq:WZ-xy}
   \ka\,\dl(x-y)\,f_{abc}\, &\Delta_c(x)(F)=\big\langle\Delta_b(y)'(F),\Delta_a(x)(F) \big\rangle-
   \big\langle\Delta_a(x)'(F),\Delta_b(y)(F) \big\rangle\nonumber\\ 
    &+\sD_a(x)\bigl(\Delta_b(y)(F)\bigr)-\sD_b(y)\bigl(\Delta_a(x)(F)\bigr)\\
    &-\big\langle\Delta_b(y)'(F),\sD_a(x)(S_0+F)\big\rangle 
    +\big\langle\Delta_a(x)'(F),\sD_b(y)(S_0+F)\big\rangle\nonumber
\end{align}
for $x,y\in g^{-1}(1)^\circ$, which is manifestly antisymmetric  under $(x,a)\leftrightarrow (y,b)$.  From 
the locality of the anomaly maps $\Dl^n$ \eqref{eq:Pnar}, 
we know that each term on the r.h.s.~has support on the diagonal $x=y$ --
in accordance with the factor $\dl(x-y)$ on the l.h.s..

\medskip

\paragraph{Restriction to the Yang-Mills interaction with FP ghosts and reduced anomaly consistency condition.}
In the following, we are going to investigate the case $F=S_\inter$ more in detail. 
In addition, we study the consistency condition \eqref{eq:WZ-xy} only for $x,y\in g^{-1}(1)^\circ$. 
That is, we think of \eqref{eq:WZ-xy} as being integrated out with test functions $X_a(x)$ and 
$Y_b(y)$ (no sum over $a,b$) with $(\supp X_a\cup\supp Y_b)\subset g^{-1}(1)^\circ$.

Our aim is to remove the anomaly $\Dl_a(x)(S_\inter)$ (for all $a$ and all $x\in g^{-1}(1)^\circ$) by finite renormalizations of 
the $R$-product and by proceeding by induction on the power of $\hbar$.
Proceeding this way, the first two terms on the r.h.s.~of  the consistency condition \eqref{eq:WZ-xy} do not 
contribute. To wit, the inductive assumption states that
\be
\Dl_a(x)(S_\inter)=\mathcal{O}(\hbar^k) \quad\text{for all $a$ and all $x\in g^{-1}(1)^\circ$.}
\ee
Since generally it holds that $\Dl_b(y)=\mathcal{O}(\hbar)$ (Thm.~\ref{thm:AMWI}(b2)), we see that
\be
\big\langle\Dl_b(y)'(S_\inter),\Dl_a(x)(S_\inter)\big\rangle=\mathcal{O}(\hbar^{k+1}).
\ee

Turning to the terms in  the last line of \eqref{eq:WZ-xy}, we take into account the \blue{crucial result 
$\sD_a(x)\,(S_0+S_\inter)=\del_\mu J_a^\mu(x)$ (Prop.~\ref{prop:dXS})} 
and that generally it holds that 
$$
\big\langle(\Dl X)'(G),\del_\mu B(x)\big\rangle=\del_\mu^x\big\langle(\Dl X)'(G), B(x)\big\rangle
\word{for all $B\in\sP$, $X\in\LieGc$ and $G\in\sF_\loc$}
$$
\blue{due to the AWI for $\Dl^n$, see Remark \ref{rem:AWI-Dl}.}
In a second step we use \blue{the explicit formula for $J^\mu$ \eqref{eq:J} and property \emph{(b6)} of 
$\Dl X$ (given in Thm.~\ref{thm:AMWI}).}
With these results, these terms can be written as
\begin{align}\label{eq:Cc-3}
-\del_\mu^x&\big\langle\Delta_b(y)'(S_\inter),\,J^\mu_a(x)\big\rangle+
\del_\mu^y\big\langle\Delta_a(x)'(S_\inter),\,J^\mu_b(y)\big\rangle\\
&=-\Bigl(\ka\,\del_\mu^x\Big\langle\Delta_b(y)'(S_\inter),\,
\bigl(\la [(\del A), A^\mu]_a(x)+[\tilde u,D^\mu u]_a(x)\bigr)\Big\rangle\Bigr)
+\Bigl((x,a)\leftrightarrow (y,b)\Bigr)\nonumber
\end{align}
for $x,y\in g^{-1}(1)^\circ$. 

Motivated by \eqref{eq:triv-anom}\blue{-\eqref{eq:J-renormal}, we 
investigate what type of terms emerge when we insert a trivial anomaly term into the ``main'' terms of the 
anomaly consistency condition:}
\begin{itemize}
\item Substituting $\del_\mu \tilde P_c^\mu(x)$ (where $\tilde P_c^\mu$ is Lorentz- and $\LieGc$-covariant and
 $\dim \tilde P^\mu_c=3$, $\dl_u(\tilde P^\mu_c)=0$, $(\tilde P^\mu_c)^*=\tilde P^\mu_c$) for 
$\Dl_c(x)(S_\inter)$ on the l.h.s.~of \eqref{eq:WZ-xy}, we obtain
\be\label{eq:P-left}
(\del_\mu^x+\del_\mu^y)\bigl(\dl(x-y)\,f_{abc}\,\tilde P_c^\mu(x)\bigr)\ ;
\ee

\item substituting $\del_\nu \tilde P_b^\nu(y)$ (with the same properties of $\tilde P^\nu_b$ as before) 
for $\Dl_b(y)(S_\inter)$ in the second line on the r.h.s.~of \eqref{eq:WZ-xy}, and using
$\frac{\dl(\del_\nu \tilde P^\nu)(y)}{\dl\vf(x)}=\del_\nu^y\frac{\dl\,\tilde P^\nu(y)}{\dl\vf(x)(x)}$ for $\vf=A_a^\mu,u_a,\tilde u_a$,
we obtain
\be\label{eq:P-right}
\sD_a(x)\bigl(\del_\nu \tilde P^\nu_b(y)\bigr)=\del_\nu^y\,\sD_a(x)\bigl(\tilde P^\nu_b(y)\bigr)=
\del_\nu^y\Bigl(\sum_{|r|\leq 3}(\del^r\dl)(x-y)\,Q^\nu_{ab,r}(y)\Bigr)
\ee
for some $Q^\nu_{ab,r}\in\sP$ being Lorentz-covariant 
having mass dimension $\dim Q^\nu_{ab,r}=3-|r|$
and satisfying $\dl_u(Q^\nu_{ab,r})=0$, $(Q^\nu_{ab,r})^*=Q^\nu_{ab,r}$, \blue{in addition, 
$Q^\nu_{ab,r}$ is such that the whole expression standing on the r.h.s.~of\eqref{eq:P-right} is $\LieG$-covariant 
(see \eqref{eq:LieG-tensor}).}
\end{itemize}

\blue{We conclude that terms in the anomaly consistency condition being of the form given on the r.h.s.~of \eqref{eq:P-right}
or the analogous form obtained by $(x,a)\leftrightarrow (y,b)$
are not of interest for our purpose, since they can be changed by addding a trivial anomaly term to $\Dl_b(y)(S_\inter)$.
(Note that \eqref{eq:P-left} belongs also to this class of terms.)
As we will see, for working out the restrictions coming from the anomaly consistency condition, 
the following observation is very helpful: also the terms \eqref{eq:Cc-3} belong to this class of 
``uninteresting'' terms; this is due to
the mentioned properties of the anomaly maps $\Dl^n$, in particular, locality \eqref{eq:Pnar} and 
homogeneous scaling \eqref{eq:Dl-homogen}.}

Due to these results, we analyze the anomaly consistency condition \eqref{eq:WZ-xy} for $F=S_\inter$
only \emph{modulo terms being of the form of the r.h.s.~of \eqref{eq:P-right} or the analogous form obtained by 
$(x,a)\leftrightarrow (y,b)$.}%
\footnote{The neglected terms are $\sim\del X_a$ or $\sim\del Y_b$ after having integrated out 
the consistency condition with $X_a(x)$ and $Y_b(y)$.}
We write ``$\simeq$'' for an equality modulo such terms;
obviously, $\simeq$ is an equivalence relation. 
\blue{(Note that by computing explicitly the l.h.s.~of \eqref{eq:P-right}, we obtain field 
polynomials $Q^\nu_{ab,r}$ which are $\LieG$-covariant by itself (since $(\tilde P^\nu_b)$ is a $\LieG$-vector);
however, for the definition of the relation $\simeq$ we only require
that the whole expression standing on the r.h.s.~of \eqref{eq:P-right} is $\LieG$-covariant.)
A crucial} advantage of this procedure is that
the terms \eqref{eq:Cc-3} are neglected, that is, in the remaining, ``reduced'' consistency condition, 
\begin{align}
\ka\,\dl(x-y)\,f_{abc}&\, \Delta_c(x)(S_\inter)\simeq \sD_a(x)\bigl(\Delta_b(y)(S_\inter)\bigr)-
 \sD_b(y)\bigl(\Delta_a(x)(S_\inter)\bigr)\label{eq:Cc-reduced}\\
 \simeq &\,\ka\, f_{acd}\Bigl(A^\mu_d(x)\,\frac{\dl\Delta_b(y)(S_\inter)}{\dl A^\mu_c(x)}+
 u_d(x)\,\frac{\dl\Delta_b(y)(S_\inter)}{\dl u_c(x)}+\tilde u_d(x)\,\frac{\dl\Delta_b(y)(S_\inter)}{\dl\tilde u_c(x)}\Bigr)\nonumber\\
&\,-\Bigl((x,a)\leftrightarrow (y,b)\Bigr)\nonumber\ ,
\end{align}
solely the anomaly $\Delta_\bullet(\bullet)(S_\inter)$ itself appears and no derivatives of it (i.e., $\Delta_\bullet(\bullet)'(S_\inter)$).
In the second $\simeq$-sign 
we have used that $\del^\mu_x\frac{\dl\Delta_b(y)(S_\inter)}{\dl A^\mu_a(x)}$ is of the form given on the r.h.s.~of
\eqref{eq:P-right} -- this follows from the results (i) \eqref{eq:form-Dl(Sint)} and (iii)-(vi) about the structure 
of $\Dl_b(y)(S_\inter)=P_b(y)$.

\begin{remk}\label{rem:WZ}
From \cite[Sect.~4]{BDFR23} \blue{(see also Sect.~3 of that reference)} we recall that for \emph{quadratic} functionals 
$F$ the ``extended Wess-Zumino consistency condition''
\eqref{eq:WZ} reduces to the original Wess-Zumino consistency condition \cite{WZ71}. More in detail, for those functionals, $\Delta X(F)$ 
is a constant functional, this follows from Thm.~\ref{thm:AMWI}(b3,b6):
\footnote{\blue{In view of the axial anomaly, note that the restriction to quadratic functionals does not exclude triangle diagrams 
without external legs; and since such diagrams correspond to constant functionals, also the pertinent
possible anomalies are constant functionals -- in agreement with \eqref{eq:anom(quadratic)}.}}
\be\label{eq:anom(quadratic)}
\fd{\Delta X(F)}{x}\overset{\eqref{eq:FI-anomaly}}{=}\Big\langle(\Dl X)'(F),\frac{\dl F}{\dl \vf(y)}\Big\rangle
\overset{\eqref{eq:Dl(vfox...)}}{=}0\ .
\ee
Then $\partial_Y\Delta X(F)$ and $\langle\Delta Y'(F),\Delta X(F)\rangle$ vanish for $X,Y\in\LieGc$
(for the latter statement see Thm.~\ref{thm:AMWI}(b1)).
That is, on the r.h.s.~of \eqref{eq:WZ}, it is solely the third line which contributes to the original Wess-Zumino consistency condition. 
But in the procedure worked out in this Sect., it is solely the second line which gives the relevant contributions; that is, 
the results derived in this Sect.~are due to the \emph{extension} of the original   Wess-Zumino condition.
\end{remk}

\paragraph{Exclusion of a non-trivial anomaly by the reduced anomaly consistency condition.}
Now, the crucial observation is the following:

\begin{lema}\label{lem:PinCc}
Since the possible anomaly term $\Delta_a(x)(S_\inter)\,t_a=P_a(x)\,t_a\equiv P(x)$ 
(with $P_a\in\sP$ and $\dim P_a=4$) \eqref{eq:form-Dl(Sint)} 
is a $\LieG$-vector (i.e., $s_a(P(x))=[t_a,P(x)]$), it satisfies the identity
\begin{align}\label{eq:P-in-cc}
\dl(x-y)\,f_{abc}\, P_c(x)&\simeq  f_{acd}\Bigl(A^\mu_d(x)\,\frac{\dl P_b(y)}{\dl A^\mu_c(x)}+
 u_d(x)\,\frac{\dl P_b(y)}{\dl u_c(x)}+\tilde u_d(x)\,\frac{\dl P_b(y)}{\dl\tilde u_c(x)}\Bigr)\nonumber\\
 &\simeq -f_{bcd}\Bigl(A^\mu_d(y)\,\frac{\dl P_a(x)}{\dl A^\mu_c(y)}+
 u_d(y)\,\frac{\dl P_a(x)}{\dl u_c(y)}+\tilde u_d(y)\,\frac{\dl P_a(x)}{\dl\tilde u_c(y)}\Bigr)\ ,
 \end{align}
where ``$\simeq$'' is defined directly before \eqref{eq:Cc-reduced}. 
 \end{lema}
 
 \begin{proof}
The second $\simeq$-sign follows immediately from the first $\simeq$-sign, since the expression 
most to the left is antisymmetric under $(a,x)\leftrightarrow (b,y)$.

To prove the first $\simeq$-sign, we write out the assumption, i.e.,  
\be\label{eq:sa(P)=[t_a,P]}
-[t_a,P(y)]=-s_a(P(y))\ .
\ee
\begin{itemize}
\item For the l.h.s.~\blue{of \eqref{eq:sa(P)=[t_a,P]}} we obtain
\be\label{eq:[ta,P]}
-[t_a,P(y)]_b\,t_b=\Bigl(\int dx\,\dl(x-y)\,f_{abc}\, P_c(x)\Bigr)\,t_b\ ,
\ee
which is the integral over $x$ of l.h.s.~of the first $\simeq$-sign \blue{in \eqref{eq:P-in-cc}} (multiplied with $t_b$). 
\item Using \eqref{eq:LieG-rot-1}, the r.h.s.~\blue{of \eqref{eq:sa(P)=[t_a,P]}} is equal to
\be\label{eq:int(rhs)}
\blue{-s_a(P(y))=}\int dx\,\,f_{acd}\Bigl(A^\mu_d(x)\,\frac{\dl P_b(y)}{\dl A^\mu_c(x)}+
 u_d(x)\,\frac{\dl P_b(y)}{\dl u_c(x)}+\tilde u_d(x)\,\frac{\dl P_b(y)}{\dl\tilde u_c(x)}\Bigr)\,t_b\ ,
\ee
which is the integral over $x$ of the r.h.s.~of the first $\simeq$-sign \blue{in \eqref{eq:P-in-cc}} (multiplied with $t_b$). 
\end{itemize}
 
 By the Poincar\'e Lemma for $\sF_\loc$-valued distributions \cite[Lemma 4.5.1]{D19} it follows the assertion. In detail:
Both, the left and the right hand side of the first $\simeq$-sign \blue{in \eqref{eq:P-in-cc}}
are linear combinations of partial derivatives of $\dl(x-y)$
and, hence, this holds also for their difference $F_{ab}(x,y):=$[l.h.s.]$-$[r.h.s.], that is,
 \be
 F_{ab}(x,y)=\sum_{|r|\leq 4}(\del^r\dl)(y-x)\,W_{ab,r}(x)\word{with} W_{ab,r}\in\sP\,,\quad\dim W_{ab,r}=4-|r|\ .
 \ee
 \blue{The equality of the expressions given on the r.h.s.~of \eqref{eq:[ta,P]} and \eqref{eq:int(rhs)}, respectively, means} that 
 $\int dx\,\,F_{ab}(x,y)=0$. Hence, we can apply the cited Lemma, which states that there exist
 field polynomials $U_{ab,r}^\mu\in\sP\,,\,\dim U_{ab,r}^\mu=3-|r|$, such that
 \be\label{eq:F=dU}
 F_{ab}(x,y)=\del_\mu^x\Bigl(\sum_{|r|\leq 4}(\del^r\dl)(y-x)\,U_{ab,r}^\mu(x)\Bigr)\ .
 \ee
 \blue{Both the left- and the right-hand side of
 the first $\simeq$-sign in \eqref{eq:P-in-cc} are $\LieG$-covariant (i.e., they satisfy \eqref{eq:LieG-tensor}
 \footnote{\blue{To verify this explicitly for the r.h.s.~of the first $\simeq$-sign in \eqref{eq:P-in-cc}, use the formula \eqref{eq:sa-fd}.}}
 ), since they are composed of $\LieG$-covariant vectors/tensors (in particular we use that $(P_b)$ is a $\LieG$-vector);
 hence, also  $F_{ab}(x,y)$ \eqref{eq:F=dU} is $\LieG$-covariant. Working manifestly Lorentz covariant in 
 construction of $U_{ab,r}^\mu$ given in the proof of the cited Lemma, $U_{ab,r}^\mu$ is a Lorentz vector.}
 That $U_{ab,r}^\mu$ can be chosen such that it satisfies $\dl_u(U^\nu_{ab,r})=0$ and
 $(U^\nu_{ab,r})^*=U^\nu_{ab,r}$ follows from \blue{the validity of these properties for $P_a$ and, hence, also for $F_{ab}(x,y)$.} 
 \end{proof}
 
{\bf Conclusions:} \emph{Since the structure constants are non-vanishing \eqref{eq:f-nonvanishing},
the reduced anomaly consistency condition \eqref{eq:Cc-reduced} excludes all possible,
non-trivial anomaly terms in the same way: $A\simeq A+A\,\,\Rightarrow\,\,A\simeq 0$.}

\emph{Due to this result, we have proved that by a finite renormalization of the gauge current $J^\mu$ \eqref{eq:J-renormal},
or equivalently a finite renormalization $R\to\wh R$ of the $R$-product given in \eqref{eq:Z}-\eqref{eq:hat(R)}, we can reach
that on-shell (i.e., restricted to configurations solving the free field equations) the interacting gauge current is conserved, i.e.,}
\be
\del_\mu^x J_{a,S_\inter}^\mu(x)\big\vert_{\sC_0}\equiv
\del_\mu^x R\Bigl(e_\ox^{S_\inter/\hbar},J_a^\mu(x)\Bigr)\Big\vert_{\sC_0}=0\word{for all}1\leq a\leq K,\,\,x\in g^{-1}(1)^\circ\ .
\ee


\section{Generality of results}\label{sec:generality}
From the procedure in Sects.~\ref{sec:gauge-trafo}-\ref{sec:ccAMWI} we recognize the following schemes: 
We consider an arbitrary model with free action $S_0$ and interaction $F\in\sF_\loc$ and a local, infinitesimal field transformation
with compact support,%
\be\label{eq:dlqQ}
\dl_{qQ}:=\int dx\,\sum_jq_j(x)Q_j(x)\,\fd{}{x}\ ,\quad q_j\in\sD(\bM,\bK)\ ,\,\,Q_j=\sum_kp_k(x)P_k(x)
\ee
with $p_k\in C^\infty(\bM,\bK)$, $P_k\in\sP$ and  $\bK=\bR$ or $\bK=\bC$. If \emph{$\dl_{qQ}$ is a (classical) symmetry of this model} 
in the sense that
\be\label{eq:dqQ=sym}
\dl_{qQ}(S_0+F)=\sum_j\int dx\,\,q_j(x)(x)\,\del_\mu J^\mu_j(x)\word{for some} J^\mu_j\in\sP\ ,
\ee
then the AMWI states that the interacting quantum currents $(J^\mu_{j,F})_j$ are conserved on-shell -- up to an anomaly:
\be\label{eq:AMWI-general}
\del_\mu^x R\Bigl(e_\ox^{F/\hbar},J_j^\mu(x)\Bigr)=R\Bigl(e_\ox^{F/\hbar},\Dl_j(x)(F)\Bigr)+\text{[terms vanishing on-shell]}\ .
\ee
for some anomaly map $\sF_\loc\ni F\mapsto\Dl_j(x)(F)\in\sP$ fulfilling Thm.~\ref{thm:AMWI}. This is 
a quantum version of \emph{Noether's Theorem} in terms of the (A)MWI.

To obtain the anomaly consistency condition of \cite{BDFR23}, we additionally need some geometric structure: 
To wit, let $\sV$ be a \emph{Lie algebra}, $(v^a)$ a basis of $\sV$ and the pertinent structure constants are denoted by 
$[v^a,v^b]=f^{ab}_{\,\,c}\,v^c$.%
\footnote{In all preceding sects.~the structure constants are assumed to be totally antisymmetric, hence all $\LieG$-indices
are written as lower indices. Here, this assumption is not made; hence, we distinguish lower and upper $\sV$-indices. Repeated indices 
are always summed over, also if both are lower indices. For example, in \eqref{eq:dX=int} there appears the $\LieG$-trace 
$\tr\bigl(X(x)\cdot\sD(x)\bigr)$.} 
In addition, let $\rho$ be some finite dimensional representation of $\sV$ with representation space $\sW$, 
and we assume that the basic field $\vf(x)$ is a functional (on the configuration space) with values in $\sW$, that is, 
$\vf(x)=\vf_j(x)\,e_j$ (sum over $j$) 
with $(e_j)$ a basis of $\sW$ and $\vf_j$ a $\bK$-valued field. We also assume that $\dl_{qQ}$ 
may be interpreted as the functional derivative in the direction of a (compactly supported)
vector field $\rho\bigl(X(x)\bigr)=X_a(x)\,\rho(v^a)$, where $X(x)=:X_a(x)\,v^a\in\sD(\bM,\sV)$ with  $X_a\in\sD(\bM,\bR)$, 
hence we write $\del_X$ instead of  $\dl_{qQ}$.  Note that \eqref{eq:dlqQ} can now be written as
\be
\del_X=\int dx\,\,(\del_X\vf)_j(x)\,\frac{\dl}{\dl\vf_j(x)}\ . 
\ee
Additionally we assume that
\be
\word{$(\del_X\vf)(x)$ is affine in $\vf(x),\del^\mu\vf(x)$}
\ee
(i.e., $(\del_X\vf)_j(x)$ is a polynomial of first order in $\vf_r(x)$, $\del^\mu\vf_s(x)$).
Writing $\del_X$ as
\be\label{eq:dX=int}
\del_X=\int dx\,\,X_a(x)\,\sD_a(x) \word{with a functional differential operator $\sD_a(x)$}
\ee
not depending on $X$, 
the above made assumption \eqref{eq:dqQ=sym} that $\dl_{qQ}$ (i.e., $\del_X$) is a symmetry 
of the model with action $(S_0+F)$ takes now the form 
\be\label{eq:dX=sym}
\sD_a(x)(S_0+F)=\del_\mu J^\mu_a(x)\word{for some} J^\mu_a\in\sP\ .
\ee
With that, the AMWI \eqref{eq:AMWI-general} can be written as
\be
\del_\mu^x (J_a^\mu)_F(x)=(\Dl_a(x)(F))_F+\text{[terms vanishing on-shell]}
\ee
for all $a$ with anomaly map $\Dl_a(x)(\bullet)$. 
If, in addition, it holds that $[\del_X,\del_Y]=\ka\,\del_{[X,Y]}$ (Lemma \ref{lem:dX=repr},
i.e., $X\mapsto \del_X$ is a \emph{Lie algebra representation}), then $\Dl_a(x)(\bullet)$ satisfies the consistency condition:
\begin{align}\label{eq:Cc-general}
 \ka\,\dl(x-y)\,f^{ab}_{\,\,c}\, \Delta_c(x)(F)=&\,\Bigl(\big\langle\Delta_b(y)'(F),\Delta_a(x)(F) \big\rangle 
+\sD_a(x)\bigl(\Delta_b(y)(F)\bigr)\\
&\,-\del_\mu^x\big\langle\Delta_b(y)'(F),J^\mu_a(x)\big\rangle\Bigr)-\Bigl((x,a)\leftrightarrow (y,b)\Bigr)\ .\nonumber
\end{align}
\blue{Removing the possible anomalies by induction on the order of $\hbar$, the first term on the r.h.s.~of \eqref{eq:Cc-general}
does not contribute. In addition, ignoring terms being $\sim\del X_a$ or $\sim\del Y_b$ after having integrated out 
the consistency condition with $X_a(x)$ and $Y_b(y)$ -- easily removable anomalies (i.e., ``trivial'' anomalies, see 
\eqref{eq:triv-anom}-\eqref{eq:J-renormal}) are of this type -- the third term on the r.h.s.~of \eqref{eq:Cc-general}
does not contribute. The resulting, reduced anomaly consistency condition,
\be\label{eq:CC-reduced}
 \ka\,\dl(x-y)\,f^{ab}_{\,\,c}\, \Delta_c(x)(F)\simeq \sD_a(x)\bigl(\Delta_b(y)(F)\bigr)-\bigl((x,a)\leftrightarrow (y,b)\bigr)\ ,
\ee
contains solely $\Dl_\bullet(\bullet)(F)$ and, hence, directly restricts the anomaly term of interest.}

In order that the reduced anomaly consistency condition \eqref{eq:CC-reduced} excludes all non-trivial, possible anomalies,
we need additional assumptions. To wit, the representation $\rho$ (in which the basic field $\vf(x)$ is) 
is the \emph{adjoint representation} of $\sV$, i.e., $\sW=\sV$ and
$\vf(x)=\vf_a(x)\,t^a$ with $(t^a)_b^{\,\,c}:=\rho(v^a)_b^{\,\,c}=f^{ac}_{\,\,\,\,b}$. In addition,
for the \emph{global} transformation belonging to $\del_X$ (which is a representation of $\sV$ on $\sF$,
see \eqref{eq:Sa}-\eqref{eq:[SP1,SP2]}), we assume that $S_0$ and $F$ are $\sV$-invariant
and that the $R$-product commutes with this global transformation 
(more precisely, it satisfies the renormalization condition \eqref{eq:[sa,R]=0},
see Prop.~\ref{pr:[Sa,R]=0}), then $\Delta(x)(F)=\Delta_c(x)(F)t_c$ is a $\sV$-vector in the adjoint representation. 
We also need that $\del_X$ 
acts on the basic field as a $\sV$-rotation up to a term $\sim\del X$, i.e.,
\be\label{eq:dXvf-gen}
\del_X\vf(x)\simeq -\ka [X(x),\vf(x)]\ .
\ee
With these additional assumptions, the reduced anomaly consistency condition \eqref{eq:CC-reduced} reads
\be\label{eq:cc-reduced}
 \dl(x-y)\,f^{ab}_{\,\,c}\, \Delta_c(x)(F)\simeq -\sum_{\vf=A^\mu,u,\tilde u}[t^a,\vf(x)]_c\frac{\dl\Dl_b(y)(F)}{\dl\vf_c(x)}
 -\bigl((x,a)\leftrightarrow (y,b)\bigr)\ ,
 \ee
 When integrating the first term on r.h.s.~over $x$ and multiplying with $t^b$ , we obtain $-s_a\bigl(\Dl(y)(F)\bigr)$ 
 (with $s_a$ defined in \eqref{eq:LieG-rot-1}).
Lemma \ref{lem:PinCc} additionally uses that the \emph{structure constants are totally antisymmetric,}%
\footnote{The validity of \eqref{eq:[ta,P]} requires this: in detail, in order that the transformation appearing on the 
l.h.s.~of \eqref{eq:cc-reduced}, i.e.,
$$
\Dl(x)(F)=\Dl_b(x)(F)\,t^b\,\,\longmapsto\,\, f^{ab}_{\,\,c}\,\Dl_c(x)(F)\,t^b\ ,
$$
is a $\sV$-rotation given by $(-t^a)$, i.e.,
$$
P(x)=P_b(x)\,t^b\,\,\overset{(-t^a)}{\longmapsto}\,\, -[t^a,P(x)]=-f^{ac}_{\,\,\,b}\,P_c(x)\,t^b\ ,
$$
it must hold that $f^{ab}_{\,\,c}=-f^{ac}_{\,\,\,b}$.} 
\blue{then it holds that $\int dx\,\mathrm{[l.h.s.~of \eqref{eq:cc-reduced}]}\,t^b=-[t^a,\Dl(y)(F)]$.
Since $\Dl(y)(F)$ is a $\sV$-vector these two integrals agree; hence,}
Lemma \ref{lem:PinCc} applies. If, in addition, the \emph{structure constants are non-vanishing} \eqref{eq:f-nonvanishing} 
(see also \eqref{eq:f-not=0-geom}), then,  the reduced consistency condition excludes all non-trivial, possible anomalies.


\section{The global transformation corresponding to the studied local gauge transformation $\del_X$}\label{sec:global-trafo}

In this section we investigate the following questions (which are essentially due to Klaus Fredenhagen): 
as we will see, the \emph{global} transformation $(\sS_a)_{a=1,\ldots,K}$ 
corresponding to the extended \emph{local} gauge transformation $\del_X$ \eqref{eq:del_X} is a classical symmetry 
(i.e., the total Lagrangian is invariant). In which sense does the pertinent Noether current $j^\mu(g;x)$ agree with the 
non-Abelian gauge current $J^\mu(x)$ \eqref{eq:J}? A main advantage of the global transformation over the local one is that 
the Lie algebra underlying $(\sS_a)$ is the Lie algebra of a
\emph{compact} Lie group; hence, the Haar meausure (on this group) is available. 
By using that measure, it is possible to symmetrize  the $R$- (or $T$-) product w.r.t.~this group, such
that this product commutes with $\sS_a$ (Prop.~\ref{pr:[Sa,R]=0}). What are the consequences of this result for 
the possible anomaly of the conservation of the interacting quantum current $(j^\mu(g;x))_{S_\inter(g)}$ 
\blue{without assuming $x\in g^{-1}(1)^\circ$} (Prop.~\ref{pr:int-anomaly})?

\paragraph{Global transformation and pertinent classical Noether current $j^\mu$.}
We study the \emph{global}, infinitesimal field transformation corresponding to $\del_X$, that is, 
in $X(x)=X_b(x)t_b\in\LieGc$ we replace $X_b(x)$ by $\dl_{ab}$.
\footnote{\blue{This section is independent of the procedure in the preceding Sections, i.e., the results are 
not obtained from the results derived in Sects.~\ref{sec:AMWI} and \ref{sec:ccAMWI} by replacing $X_b(x)$ by $\dl_{ab}$.}}
This yields
\be\label{eq:Sa}
\sS_a\,:\,\sF\longrightarrow\sF\,;\,\sS_a=\int dx\,\,\sD_a(x)\ .
\ee
Inserting the explicit formula \eqref{eq:dXF} for $\sD_a(x)$, we recognize that the first term, i.e., the contribution of the 
infinitesimal field shift to $\del_X^A$, does not contribute since the integration runs only over the support of $F$ (which is compact).
Therefore, $\sS_a$ is solely a $\LieG$-rotation:
\begin{align}\label{eq:Sa-explicit}
\sS_a(F)=\,&\,-\ka\int dx\,\Bigl(\Bigl[A^\mu(x),\frac{\dl F}{\dl A^\mu (x)}\Bigr]_a
+\Bigl[u(x),\frac{\dl_l F}{\dl u(x)}\Bigr]_a+\Bigl[\tilde u(x),\frac{\dl_l F}{\tilde u(x)}\Bigr]_a\Bigr)\nonumber\\
=\,&\,-\ka\int dx\,\Bigl([t_a,A^\mu(x)]_b\frac{\dl F}{\dl A_b^\mu (x)}
+[t_a,u(x)]_b\frac{\dl_l F}{\dl u_b(x)}+[t_a,\tilde u(x)]_b\frac{\dl_l F}{\tilde u_b(x)}\Bigr)\nonumber\\
=\,&\,-\ka\,s_a(F)
\end{align}
with $s_a$ defined in \eqref{eq:LieG-rot-1}. Note that $\om_0\circ\sS_a=0$ \blue{with $\om_0$ being the vacuum state
\eqref{eq:vacuum}. To wit, for all basic fields $\vf$ it holds that 
$\om_0\bigl([t_a,\vf(x)]\cdot\fd{F}{x}\bigr)=\om_0([t_a,\vf(x)])\cdot\om_0\bigl(\fd{F}{x}\bigr)$ and $\om_0([t_a,\vf(x)])=0$.} 
We also point out that
\be\label{eq:P->SP} 
\LieG\ni P=P_aT_a\longmapsto S_P:=P_a\,\sS_a\word{is a representation of $\LieG$ on $\sF$.}
\ee 
To wit, using the relations $[\del_X,\del_Y]=\ka\,\del_{[X,Y]}$ \eqref{eq:[dX,dY]} and 
$$
\del_{[X,Y]}=\int dxdy\,\,\dl(x-y)\,X_a(x)Y_b(y)\,f_{abc}\,\sD_c(x)
$$ 
and going over to the corresponding global transformations in the former relation (as explained before \eqref{eq:Sa}), we obtain%
\footnote{Analogously to \eqref{eq:[dX,dY]}, $[\sS_a,\sS_b]$ denotes the commutator of two functional differential operators, 
and similarly for $[\sS_{P_1},\sS_{P_2}]$. \blue{The identity \eqref{eq:[SP1,SP2]} can also be obtained directly from 
\eqref{eq:Sa-explicit} by a reduced version of the proof of Lemma \ref{lem:dX=repr}.}}
\be\label{eq:[SP1,SP2]}
[\sS_a,\sS_b]=\ka\,f_{abc}\,\sS_c\ ,\word{that is,} [\sS_{P_1},\sS_{P_2}]=\ka\,\sS_{[P_1,P_2]}\ .
\ee 
Obviously, $\sS_a$ is a derivation w.r.t.~the classical product  and, due to
$\frac{\dl(\del_\nu B)(x)}{\dl \vf(y)}=\del_\nu^x\Bigl(\frac{\dl B(x)}{\vf(y)}\Bigr)$, $ \vf=A^\mu_a,u_a,\tilde u_a$,
it holds that
\be
\sS_a\bigl(\del^\mu B(x)\bigr)=\del^\mu_x\,\sS_a\bigl(B(x)\bigr)\word{for $B\in\sP$.}
\ee

That $\sS_a$ is a derivation also with respect to the star-product \eqref{eq:star-product} relies on the cancelations appearing in
the commutator
\begin{align}
&[\sD\,,\,\sS_a\ox\id+\id\ox\sS_a]=\ka\int dx\,dy\,\Bigl(D^{\mu\nu,+}_\la(x-y)\,
\frac{\dl}{\dl A_{b}^\mu(x)}\ox\frac{\dl}{\dl A^\nu_c(y)}\nonumber\\
&+D^+(x-y)\Bigl(-\frac{\dl_r}{\dl u_b(x)}\ox\frac{\dl_l}{\dl \tilde u_c(y)}+\frac{\dl_r}{\dl \tilde u_b(x)}\ox\frac{\dl_l}{\dl u_c(y)}\Bigr)\Bigr)
\cdot(f_{abc}+f_{acb})=0\ ;
\end{align} 
hence, we obtain
\begin{align}\label{eq:Sa-starproduct}
\sS_a(F\star G)=&\,\sS_a\circ\sM\circ e^{\hbar\sD}(F\ox G)=\sM\circ (\sS_a\ox\id+\id\ox\sS_a)\circ e^{\hbar\sD}(F\ox G)\nonumber\\
=&\,\sM\circ e^{\hbar\sD}\bigl(\sS_a(F)\ox G+F\ox\sS_a(G)\bigr)=\sS_a(F)\star G+F\star\sS_a(G)\ .
\end{align}

In \eqref{eq:sa(tr)}-\eqref{eq:sa(LYM)} we already verified that 
\be
\sS_a\bigl(L_\YM(g;x)\bigr)=0\ ,\quad \sS_a\bigl(L_\gf(g;x)\bigr)=0\ ,\quad \sS_a\bigl(L_\gh(g;x)\bigr)=0
\ee
\blue{(remember $\sS_a=-\ka\,s_a$)};
hence, the total Lagrangian $L(g;x)=L_\YM(g;x)+L_\gf(g;x)+L_\gh(g;x)$ is invariant.
At this stage, $F^{\mu\nu}$ and $D^\mu$ contain the switching function $g$, hence, this holds also for the Lagrangians; 
we indicate this by writing $L(g;\bullet)$. We derive the pertinent classical Noether current $j^\mu_a$ 
in the usual way (cf.~e.g.~\cite[Chap.~4.2.3]{D19}):
\begin{align}\label{eq:Noether-Sa}
0={}&\sS_a\bigl(L(g;x)\bigr)=\sum_k \sS_a\bigl(\vf_k(x)\bigr)\cdot\frac{\del L(g;\bullet)}{\del\vf_k}(x)+
\sS_a\bigl(\del_\mu\vf_k(x)\bigr)\cdot\frac{\del L(g;\bullet)}{\del(\del_\mu\vf_k)}(x)
\\
={}& \del^x_\mu \sum_k\Bigl(\sS_a\bigl(\vf_k(x)\bigr)\cdot\frac{\del L(g;\bullet)}{\del(\del_\mu\vf_k)}(x)\Bigr)
+\sum_k \sS_a\bigl(\vf_k(x)\bigr)\cdot\Bigl[\frac{\del L(g;\bullet)}{\del\vf_k}(x)-
\del_\mu\frac{\del L(g;\bullet)}{\del(\del_\mu\vf_k)}(x)\Bigr]\ ,\nonumber
\end{align}
where the sum over $k$ is the sum over $\vf_k=A_b^\mu,\,u_b,\,\tilde u_b$. Therefore, in classical field theory
\blue{the interacting field belonging to} the Noether current
\be\label{eq:current-Sa}
j_a^\mu(g;x):=-\Bigl(\sS_a\bigl(A^\nu_b(x)\bigr)\frac{\del L(g;\bullet)}{\del(\del_\mu A^\nu_b)}(x)
+\sS_a\bigl(u_b(x)\bigr)\frac{\del L(g;\bullet)}{\del(\del_\mu u_b)}(x)+
\sS_a\bigl(\tilde u_b(x)\bigr)\frac{\del L(g;\bullet)}{\del(\del_\mu\tilde u_b)}(x)\Bigr)
\ee
is conserved modulo the interacting field equations. \blue{(Note that \eqref{eq:Noether-Sa} and \eqref{eq:current-Sa}
are equations for field polynomials with coefficients containing the function $g$.)}
In terms of the perturbative, classical, retarded fields
\eqref{eq:class-int-field} this result reads%
\footnote{We recall that the lower index `$0$' denotes restriction of the functional to the space $\sC_0$
of solutions of the free field equations.}
\be
\del_\mu^x j_a^\mu(g;x)^\ret_{S_\inter(g),0}=0\ .
\ee
Computing $j^\mu(g;x):=j^\mu_a(g;x)\,t_a$ for the model at hand, we obtain
\be\label{eq:j}
j^\mu(g;x)=\ka\Bigl(-[A_\nu(x),F^{\mu\nu}(x)]+\la[(\del A)(x),A^\mu(x)]-[\del^\mu\tilde u(x),u(x)]+
[\tilde u(x),D^\mu u(x)]\Bigr);
\ee
this result differs from the non-Abelian gauge current $J^\mu$ given in \eqref{eq:J}\blue{: the first and the third term do 
not appear in $J^\mu$, and $J^\mu$ has a term $\la\square A^\mu$ which does not emerge in $j^\mu$.}  
Setting $g(x)=0$, we obtain the corresponding current for the free theory,
\begin{align}\label{eq:j-g=0}
j^\mu(0;x)=\ka\Bigl(&-[A_\nu(x),(\del^\mu A^\nu(x)-\del^\nu A^\mu(x))]+\la[(\del A)(x),A^\mu(x)]\nonumber\\
&-[\del^\mu\tilde u(x),u(x)]+[\tilde u(x),\del^\mu u(x)]\Bigr),
\end{align}
which is conserved modulo the free field equations, i.e., 
\be\label{eq:dj(0)=0}
\del_\mu^x j^\mu(0;x)_0=0\ .
\ee

\paragraph{In which sense agrees $j^\mu$ with the non-Abelian gauge current $J^\mu$?}
This paragraph applies to both (perturbative) classial and quantum field theory; for the interacting currents we use the notation of pQFT.
To study the difference $(J^\mu-j^\mu)$ we restrict to $x\in g^{-1}(1)^\circ$ and we write
$j^\mu(x):=j^\mu(g;x)$ if $(x,g)$ satisfies this assumption. The crucial point is that
\emph{on-shell the difference $(J^\mu_{S_\inter,0}-j^\mu_{S_\inter,0})$ is a term whose divergence vanishes identically}.
To wit, by using the field equation \eqref{eq:FE-A},
\footnote{\blue{In pQFT the validity of the field equations is a renormalization condition for the $R$-product, see \eqref{eq:FE}.}}
the first term of $J^\mu_{S_\inter,0}$ can be written as
\begin{align}
\la \square A^\mu_{S_\inter,0}(x)={}&\square A^\mu_{S_\inter,0}(x)-(1-\la)\square A^\mu_{S_\inter,0}(x)\\
={}&(1-\la)\del^\mu\del^\nu (A_\nu)_{S_\inter,0}(x)+\ka\Bigl(\del^x_\nu([A^\nu(x),A^\mu(x)])_{S_\inter,0}\nonumber\\
&-([A_\nu(x),F^{\mu\nu}(x)])_{S_\inter,0}-([\del^\mu\tilde u(x),u(x)])_{S_\inter,0}\Bigr)-(1-\la)\square A^\mu_{S_\inter,0}(x)\ .
\nonumber
\end{align}
Inserting this result we obtain
\be\label{eq:J-j}
J^\mu_{S_\inter,0}(x)-j^\mu_{S_\inter,0}(x)=(1-\la)(\del^\mu\del^\nu-g^{\mu\nu}\square) (A_\nu)_{S_\inter,0}(x)
+\ka\,\del^x_\nu([A^\nu(x),A^\mu(x)])_{S_\inter,0}\ ;
\ee
obviously, the divergence $\del_\mu^x$[r.h.s.~of \eqref{eq:J-j}] vanishes identically, i.e., without restriction to $\sC_0$.
In pQFT the conclusion is that \emph{on-shell conservation of $j^\mu$ and $J^\mu$ is violated by the same anomaly term:}
\be\label{eq:dj=dJ}
 \del_\mu^x\,j^\mu_{S_\inter,0}(x)=\del_\mu^x\,J^\mu_{S_\inter,0}(x)=\bigl(\Delta(x)(S_\inter)\bigr)_{S_\inter,0}
 \word{for $x\in g^{-1}(1)^\circ$,}
\ee
where $\Delta(x)(S_\inter)=\Delta_a(x)(S_\inter)\,t_a$.

\paragraph{How to fulfil invariance of the $R$-product w.r.t.~the global transformation, i.e.,
the renormalization condition $\LieG$-covariance?}
A crucial advantage of $\sS_a$ over $\del_X$ is that $\sS_a$ is a \emph{global} transformation.
In addition, since the structure constants are totally antisymmetric and non-vanishing, the Killing form is negative definite. In detail,
the Killing form is defined  by 
\be
\langle T_a,T_b\rangle:=\mathrm{Tr}(t_at_b)=\sum_{c,d}f_{adc}f_{bcd}
\ee
(where $\mathrm{Tr}(t_at_b)$ is the matrix-trace of the matrix product $t_a\cdot t_b$) and, by the assumed properties of 
the structure constants (see the beginning of Sect.~\ref{sec:YM-basics}), it holds that
\be
\langle T_a,T_a\rangle=-\sum_{c,d}(f_{adc})^2<0\ .
\ee
Since the Killing form is negative definite, $\LieG$ is the Lie algebra of a semisimple compact Lie group $\G$. Due to the compactness,
the Haar measure is available. By using the latter a quite simple proof can be given of the following statement.

\begin{prop}\label{pr:[Sa,R]=0}
In the inductive Epstein-Glaser construction of the sequence $(R_{n,1})_{n\in\bN}$ (see \cite{DF04} or \cite[Chaps.~3.1,3.2]{D19}), 
the symmetry relation
\be\label{eq:[Sa,R]=0}
\sS_a\circ R_{n-1,1}=R_{n-1,1}\circ\sum_{k=1}^{n}(\id\oxyox\sS_a\oxyox\id)\quad\forall a
\ee
(where on the r.h.s.~$\sS_a$ is the $k$th factor) is a renormalization condition, which can be 
satisfied by a symmetrization, which maintains the validity of all other renormalization conditions.
 \end{prop}
 
 \begin{proof} That \eqref{eq:[Sa,R]=0} is a renormalization condition is most easily seen by considering the analogous (and equivalent)
 relation for $T_n$ and by proceeding by induction on $n$, following the inductive Epstein-Glaser construction of the sequence $(T_n)$
 (\cite{EG73} or\cite[Chap.~3.3]{D19}). The claim follows then from the above obtained result that $\sS_a$ is a derivation 
 w.r.t.~the star product \eqref{eq:Sa-starproduct}. 
 
 With that, in the inductive step $(n-1)\to n$, we know that $T_n^0(\ldots)\in\sD'(\bM^n\setminus\Dl_n,\sF)$ 
 (where $\Dl_n:=\{(x_1,\ldots,x_n\in\bM^n\,\big\vert\,x_1=\ldots =x_n\}$)
 \footnote{Here we understand $T_n$ as a map $\sP^{\ox n}\to \sD'(\bM^n,\sF)$ and analogously for $T_n^0$,
 see e.g.~\cite[Chap.~3.1.1]{D19}.} 
  satisfies \eqref{eq:[Sa,R]=0}. Let $T_n$ be an ``admissible'' extension of $T_n^0$ to $\sD'(\bM^n,\sF)$, i.e., an extension fulfilling all further renormalization conditions. We aim to construct
 a symmetrization $T_n^\sym$ of $T_n$ which satisfies \eqref{eq:[Sa,R]=0} and maintains all other renormalization conditions.
 We do this by means of the Haar measure. Since this is a measure on the Lie group (and not on the Lie algebra),
 we study the \emph{finite} transformation belonging to $(\sS_a)_a$, more precisely,
 \be
 \g\equiv\g(\unl\la):=\exp(\unl\la\,\unl\sS)\word{with}\unl\la\,\unl\sS:=\sum_a\la_a\sS_a\ .
 \ee 
 Defining 
 \be
 V(\g)T_n:=\g\circ T_n\circ (\g^{-1})^{\ox n}\ ,
 \ee
 the symmetrized extension must satisfy
 \be\label{eq:V(g)T=T}
 V(\g)T_n^\sym=T_n^\sym\ ,
 \ee
 which  is \emph{equivalent} to the assertion \eqref{eq:[Sa,R]=0} for $T_n^\sym$ in place of $R_{n-1,1}$.
 To verify this equivalence, first note that application of $\frac{\del}{\del \la_a}\big\vert_{\unl\la=\unl 0}$ to 
$\g\circ T_n^\sym=T_n^\sym\circ \g^{\ox n}$ \eqref{eq:V(g)T=T}  
 yields \eqref{eq:[Sa,R]=0} for $T_n^\sym$. The reversed conclusion is obtained by multiple application of \eqref{eq:[Sa,R]=0}:
 \begin{align*}
\exp & (\unl\la\,\unl\sS)\,T_n^\sym=\sum_{j=0}^\infty\frac{(\unl\la\,\unl\sS)^j}{j!}\,T_n^\sym\overset{\eqref{eq:[Sa,R]=0}}{=}
T_n^\sym\circ\sum_{j=0}^\infty\frac{\Bigl(\sum_{k=1}^{n}(\id\oxyox\unl\la\,\unl\sS\oxyox\id)\Bigr)^j}{j!}\\
=&\,T_n^\sym\circ\sum_{j=0}^\infty\frac{1}{j!}\sum_{j_1+\ldots+j_n=j}\frac{j!}{j_1!\cdots j_n!}\,
(\unl\la\,\unl\sS\ox\id\oxyox\id)^{j_1}\cdots(\id\oxyox\id\ox\unl\la\,\unl\sS)^{j_n}\\
=&\,T_n^\sym\circ\sum_{j_1,\ldots,j_n=0}^\infty\frac{(\unl\la\,\unl\sS)^{j_1}}{j_1!}\oxyox\frac{(\unl\la\,\unl\sS)^{j_n}}{j_n!}\\
=&\,T_n^\sym\circ\Bigl(\exp\bigl(\unl\la\,\unl\sS\bigr)\Bigr)^{\ox n}\ ,
\end{align*} 
where the multinomial theorem is used.
 
 With that we can follow \cite[App.~D]{DF04} or\cite[Chap.~3.2.7]{D19}.
 Since \eqref{eq:[Sa,R]=0} is a renormalization condition, we know that
 \be
 V(\g)T_n^0=T_n^0\ .
 \ee
 Due to $V(\g_1\g_2)=V(\g_1)\circ V(\g_2)$, the map $\g\mapsto V(\g)$ is a representation. 
 
 Let an admissble extension $T_n$ of $T_n^0$ be given. One verifies that $V(\g)T_n$ is then also an admissible  extension of
  $V(\g)T_n^0=T_n^0$. Hence,
 \be
 L_n(\g):=V(\g)T_n-T_n 
 \ee
 is the difference of two admissible extensions, in particular it fulfills $\supp L_n(\g)\subset\Dl_n$.
 One straightforwardly verifies that the map $\g\mapsto L_n(\g)$ satisfies the ``cocycle'' relation
 \be\label{eq:cocycle}
 L_n(\g_1\g_2)=V(\g_1)\circ L_n(\g_2)+L_n(\g_1) \ .
 \ee
 
 We are searching an $L_n^\sym$ such that $T_n^\sym:=T_n+L_n^\sym$ is an admissible extension of $T_n^0$ (in particular, 
 it must hold that $\supp L_n^\sym\subset\Dl_n$) and that $T_n+L_n^\sym$ is invariant,
 \be
 T_n+L_n^\sym\overset{!}{=}V(\g)\bigl(T_n+L_n^\sym\bigr)=L_n(\g)+T_n+V(\g)L_n^\sym\ ,
 \ee
 which is equivalent to
 \be\label{eq:coboundary}
 L_n(\g)=L_n^\sym-V(\g)L^\sym_n\ .
 \ee
 We claim that $L_n^\sym$ is obtained by the symmetrization
 \be\label{eq:symmetrization}
 L_n^\sym:=\int_{\G_0}d\g\,\,L_n(\g)\ ,
 \ee
 where $\G_0$ is the connected component of the unit $\bf{1}=\g(\unl 0)\in\G$ and $d\g$ is the 
 Haar measure, that is, the uniquely determined measure on
 $\G_0$ which is invariant under left- and right-translations (i.e., $d(\h\g)=d\g=d(\g\h)$ for $\h\in\G_0$) and has norm $1$ 
 (i.e., $\int_{\G_0}d\g={\bf 1}$). We have to verify that  $L_n^\sym$ (given by \eqref{eq:symmetrization}) satsifies \eqref{eq:coboundary}:
 \begin{align*}
  L_n^\sym -V(\g)L^\sym_n\overset{\eqref{eq:symmetrization}}{=}{}&\int_{\G_0} d\h\,\bigl(L_n(\h)-V(\g)\circ L_n(\h)\bigr)\\
 \overset{\eqref{eq:cocycle}}{=}{}&\int_{\G_0} d\h\,\bigl(L_n(\h)-L_n(\g\h)+L_n(\g)\bigr)\\
 ={}&\int_{\G_0} d\h\,L_n(\h)-\int_{\G_0} d(\g\h)\,L_n(\g\h)+\Bigl(\int_{\G_0} d\h\Bigr)\,L_n(\g)\\
 ={}&L_n(\g)\ .
 \end{align*}
 That $T_n^\sym=T_n+L_n^\sym$ is an admissible extension of $T_n^0$ is a consequence of \eqref{eq:symmetrization}
 and the (above obtained) result that $T_n+L_n(\g)$ is an admissible extension of $T_n^0$ for all $\g\in\G_0$. 
 \end{proof}
 
 \paragraph{The possible anomaly of $\del_\mu^x\,j^\mu(g;x)_{S_\inter,0}$ without assuming  $x\in g^{-1}(1)^\circ$.}
In the remainder of this section, we study conservation of $j^\mu(g;x)_{S_\inter,0}$ for general $(g;x)$.
\blue{Inserting \eqref{eq:current-Sa} into \eqref{eq:Noether-Sa}} we obtain
\be\label{eq:dj}
\del_\mu^x\,j^\mu(g;x)=\sum_k\sS_a\bigl(\vf_k(x)\bigr)\cdot \frac{\dl\bigl(S_0+S_\inter(g)\bigr)}{\dl\vf_k(x)}\ ,
\ee
where the sum over $k$ is the sum over $\vf_k=A_b^\mu,\,u_b,\,\tilde u_b$. \blue{(Similarly to 
\eqref{eq:Noether-Sa}-\eqref{eq:current-Sa} this is an equation for field polynomials.)}
Hence we can apply the theorem about the AMWI (given in \cite[Sect.~5.2]{Brennecke08} or \cite[Chap.~4.3]{D19}, 
\blue{see also Thm.~\ref{thm:AMWI} and \eqref{eq:Dla(x)}}): there exists
a unique local functional
\be
\tilde\Dl(x)(S_\inter(g))=\tilde\Dl_a(x)(S_\inter(g))\,t_a\ ,\quad\tilde\Dl_a(x)(S_\inter(g))\in\sF_\loc\ ,
\ee
which satisfies certain properties (recalled in Thm.~\ref{thm:AMWI} for the case that the infinitesimal symmetry transformation is 
$\del_X$): most important are the AMWI, whose on-shell version reads
\be\label{eq:AMWI-j-onshell}
\del_\mu^x\,\bigl(j^\mu(g;x)\bigr)_{S_\inter(g),0}=\bigl(\tilde\Dl(x)(S_\inter(g))\bigr)_{S_\inter(g),0}\ ,
\ee
and locality, see \eqref{eq:Pnar}. \blue{Due to the latter, the map $\Dl_a(x)(\bullet)$ introduced in \eqref{eq:Dla(x)} satisfies
$\Dl_a(x)\bigl(B(g)\bigr)=0$, if $x\not\in\supp g$. For the possible anomaly $\tilde\Dl(x)(S_\inter(g))$ this means that}
\be\label{eq:Dl=0}
\tilde\Dl(x)(S_\inter(g))=0\word{for $x\not\in\supp g$.}
\ee

\begin{prop}\label{pr:int-anomaly}
 If the $R$-product commutes with the global transformation $\sS_a$ (more precisely, \eqref{eq:[Sa,R]=0} is satisfied),
then it holds that
\be\label{eq:int-dj(g;x)=0}
\int dy \,\,\bigl(\tilde\Dl(y)(S_\inter(g))\bigr)_{S_\inter(g),0}=0\ .
\ee
\end{prop}

The proof of this proposition uses the following result:
\begin{lema}\label{lem:[Q,F]=S_a(F)}
It holds that
\be\label{eq:[Q,F]=S_a(F)}
\int d^3y\,\,[j^0_a\bigl(0;(c,\vec y)\bigr)_0\,,\,F_0]_\star =i\hbar\,\sS_a(F)_0\ ,\quad F\in\sF\ ,
\ee
for any choice of $c\in\bR$, where $[\bullet,\bullet]_\star$ denotes the commutator w.r.t.~the on-shell star product 
\eqref{eq:on-shell-star-product}.
\end{lema}

\emph{Proof of Prop.~\ref{pr:int-anomaly}.} Let $\sO\subset \bM$
be an open double cone containing $\supp g$, and let $g_1\in\sD(\bM)$
be a test function which is equal to $1$ on a neighbourhood of $\ovl\sO$.
With that and \eqref{eq:AMWI-j-onshell}-\eqref{eq:Dl=0}, \eqref{eq:on-shell-star-product} and the Bogoliubov formula
\eqref{eq:Bogoliubov}, we may write
\begin{align}\label{eq:int-anomaly-1}
\int dy\,\,\bigl(\tilde\Dl(y)(&S_\inter(g))\bigr)_{S_\inter(g),0}=\int dy\,\,g_1(y)\,\bigl(\tilde\Dl(y)(S_\inter(g))\bigr)_{S_\inter(g),0}\\
={}&-\ovl T\bigl(e_\ox^{-i\,S_\inter(g)/\hbar}\bigr)_0\star\int dy\,(\del^\mu g_1)(y)\, 
T\bigl(e_\ox^{i\,S_\inter(g)/\hbar}\ox j^\mu(g;y)\bigr)_0\ .
\nonumber
\end{align}

\blue{\begin{figure}[ht]
\centering
\begin{tikzpicture}[scale=0.7]
\coordinate (A) at (-4,0) ; \coordinate (B) at (0,-4) ;
\coordinate (C) at (4,0) ; \coordinate (D) at (0,4) ;
\draw (A) -- (B) -- (C) -- (D) -- (A); 
\coordinate (E) at (-2.5,0) ; \coordinate (F) at (0,-2.5) ;
\coordinate (G) at (2.5,0) ; \coordinate (H) at (0,2.5) ;
\draw (E) to[out=-90, in=180] (F) to[out=0, in=-90] (G) to[out=90, in=0] (H) to[out=180, in=90] (E);
\fill[black!15] (E) to[out=-90, in=180] (F) to[out=0, in=-90] (G) to[out=90, in=0] (H) to[out=180, in=90] (E);
\coordinate (I) at (-6,0.5) ; \coordinate (J) at (-5,-0.5) ; 
\coordinate (K) at (-0.5,4) ;\coordinate (L) at (0.5,4) ; 
\coordinate (M) at (5,-0.5) ;\coordinate (N) at (6,0.5) ; 
\coordinate (O) at (0.5,6) ;\coordinate (P) at (-0.5,6) ; 
\coordinate (Q) at (-6.2,0) ;\coordinate (R) at (6.2,0) ; 
\coordinate (S) at (-7,3) ; \coordinate (T) at (-7,-3) ;
\coordinate (U) at (7,-3) ; \coordinate (V) at (7,3) ;
\draw[dashed] (A) -- (S) ; \draw[dashed] (A) -- (T) ; \draw[dashed] (C) -- (U) ; \draw[dashed] (C) -- (V) ; 
\draw[dashed] (I) to[out=-135, in=90] (Q) to[out=-90, in=-135]
(J) -- (K) to[out=45, in=135] (L) --(M) to[out=-45, in=-90] (R) to[out=90, in=-45] (N) -- (O) to[out=135, in=45] (P) -- (I);
\fill[black!5]  (I) to[out=-135, in=90] (Q) to[out=-90, in=-135] (J) -- (K) to[out=45, in=135] (L) --
(M) to[out=-45, in=-90] (R) to[out=90, in=-45] (N) -- (O) to[out=135, in=45] (P) -- (I);
\coordinate (I1) at (-6,-0.5) ; \coordinate (J1) at (-5,0.5) ; 
\coordinate (K1) at (-0.5,-4) ;\coordinate (L1) at (0.5,-4) ; 
\coordinate (M1) at (5,0.5) ;\coordinate (N1) at (6,-0.5) ; 
\coordinate (O1) at (0.5,-6) ;\coordinate (P1) at (-0.5,-6) ; 
\draw[dashed] (I1) to[out=135, in=-90] (Q) to[out=90, in=135]  (J1) -- (K1) to[out=-45, in=-135] (L1) --(M1) 
 to[out=45, in=90] (R) to[out=-90, in=45]  (N1) -- (O1) to[out=-135, in=-45] (P1) -- (I1);
\fill[black!5]  (I1) to[out=135, in=-90] (Q) to[out=90, in=135]  (J1) -- (K1) to[out=-45, in=-135] (L1) --
(M1) to[out=45, in=90] (R) to[out=-90, in=45] (N1) -- (O1) to[out=-135, in=-45] (P1) -- (I1);
\coordinate (X) at (0,-3.4) ; \coordinate (Y) at (0,-2) ;
\draw (X)  node {$\sO$};\draw (Y)  node {$\supp g$};
\coordinate (Z) at (0,-5) ; \coordinate (Z1) at (0,5) ;
\draw (Z)  node {$\supp b^\mu$};\draw (Z1)  node {$\supp a^\mu$};
\draw (V)  node[left] {$\sO+\ovl V_+$};\draw (U)  node[left] {$\sO+\ovl V_-$};
\end{tikzpicture}
\caption{\small\blue{Illustration of the proof of Prop.~\ref{pr:int-anomaly}: double cone $\sO$ containing $\supp g$ (with its past $\sO+\ovl V_-$
and its future $\sO+\ovl V_+$) and respective supports of the (decomposed) test functions.}}
\label{fg:TSA} 
\end{figure}}

\noindent We decompose $\del^{\mu} g_1=a^{\mu}-b^{\mu}$
such that $\supp a^{\mu}\cap (\sO+\overline{V}_-)=\emptyset$ and
$\supp b^{\mu}\cap (\sO+\overline{V}_+)=\emptyset$. By causal
factorization of the $T$-product and \eqref{eq:on-shell-star-product} and $j^\mu(g;a^\mu)=j^{\mu}(0;a_{\mu})$
(since $\supp g\cap\supp a^\mu=\emptyset$), and similarly for $j^\mu(g;b^\mu)$, we obtain
\begin{align}
\int dy\,&(\del^\mu g_1)(y)\, 
T\bigl(e_\ox^{i\,S_\inter(g)/\hbar}\ox j^\mu(g;y)\bigr)_0\nonumber\\
&=j^{\mu}(0;a_{\mu})_0\star T\bigl(e_\ox^{i\,S_\inter(g)/\hbar}\bigr)_0-
T\bigl(e_\ox^{i\,S_\inter(g)/\hbar}\bigr)_0\star j^{\mu}(0;b_{\mu})_0\nonumber\\
&=[j^{\mu}(0;a_{\mu})_0\,,\,T\bigl(e_\ox^{i\,S_\inter(g)/\hbar}\bigr)_0]_\star
+T\bigl(e_\ox^{i\,S_\inter(g)/\hbar}\bigr)_0\star j^{\mu}(0;\partial_{\mu}g)_0\ .\label{eq:int-anomaly-2}
\end{align}
In the last term the second factor vanishes, due to \eqref{eq:dj(0)=0}. From Field
independence of $T$ we know that
$\,\supp T\bigl(e_\ox^{i\,S_\inter(g)/\hbar}\bigr)_0\subset{\cal O}$;
therefore, we may vary $a^\mu$ in the spacelike complement of $\sO$
without affecting the remaining commutator in \eqref{eq:int-anomaly-2}.
Taking additionally into account that $g_1$ is arbitrary
(up to $g_1\big\vert_{\ovl\sO}=1$),
we may choose for $a_\mu(y)$ a smooth approximation to $-\dl_{\mu 0}\,\dl(y^0-c)$,
where $c\in \bR$ is a sufficiently large constant:
$$
a_\mu(y)=-\dl_{\mu 0}\,h(y^0)\word{with} \int dy^0\,\, h(y^0)=1\ ,\quad
h\in\sD([c-\eps,c+\eps])
$$
for some $\eps >0$. Inserting this $a^\mu$ into \eqref{eq:int-anomaly-2} and using Lemma \ref{lem:[Q,F]=S_a(F)}, we get
\begin{align*}
[j^{\mu}(0;a_{\mu})_0\,,\,T\bigl(e_\ox^{i\,S_\inter(g)/\hbar}\bigr)_0]_\star={}& -\int dy^0\,\,h(y^0)\int d\vec y\,\,
[j^0(0;(y^0,\vec y))_0\,,\,T\bigl(e_\ox^{i\,S_\inter(g)/\hbar}\bigr)_0]_\star \\
={}& -\Bigl(\int dy^0\,\,h(y^0)\Bigr)\int d\vec y\,\,
[j^0(0;(c,\vec y))_0\,,\,T\bigl(e_\ox^{i\,S_\inter(g)/\hbar}\bigr)_0]_\star \\
={}&-i\hbar\,\sS_a\Bigl(T\bigl(e_\ox^{i\,S_\inter(g)/\hbar}\bigr)\Bigr)_0\ .
\end{align*}
By the validity of \eqref{eq:[Sa,R]=0} for $T_n$ and by $\sS_a\bigl(S_\inter(g)\bigr)=0$ we obtain the assertion \eqref{eq:int-dj(g;x)=0}.
$\qed$

 \emph{Proof of Lemma \ref{lem:[Q,F]=S_a(F)}.} 
 Due to spacelike commutativity the region of integration on the l.h.s.~of \eqref{eq:[Q,F]=S_a(F)} is a subset of
$$
\set{\vec y\,\big\vert\,(c,\vec y)\in\supp F+(\ovl V_+\cup\ovl V_-)}\ ,
$$
which is bounded for all $c\in\bR$; therefore, this integral exists indeed. In addition, due to \eqref{eq:dj(0)=0}
and Gauss' integral theorem, this integral does not depend on $c$. \blue{Hence, we may introduce the short notation
\footnote{\blue{For a rigorous definition of the pertinent charge $Q_0$, 
we refer to \cite{DF99} or \cite[Chap.~5.5.1]{D19} -- however, here we solely deal with the commutator 
\eqref{eq:[Q,F]}, which is well-defined as it is written here.}}
\be\label{eq:[Q,F]}
[Q_0,F_0]_\star:=\int d^3y\,[j^0_a\bigl(0;(c,\vec y)\bigr)_0\,,\,F_0]_\star\ ,\quad\forall F\in\sF\ .
\ee}

We first prove \eqref{eq:[Q,F]=S_a(F)} for $F=A_b^\rho(x)$: Inserting \eqref{eq:j-g=0} and using
$[A^\mu_a(y)\,,\,A^\rho_b(x)]_\star =-i\hbar\dl_{ab}\,D_\la^{\mu\rho}(y-x)$ we obtain
\begin{align}\label{eq:[Q,A]}
 \int d^3y\,\,[j^0_a&(0;y)_0\,,\,A_b^\rho(x)_0]_\star =i\hbar\ka\,f_{acd}\int d^3y\,
 \Bigl(\dl_{cb}\,\bigl(\del^0 A^\nu_d(y)-\del^\nu A^0_d(y)\bigr)\,D_{\nu,\,\la}^{\,\,\rho}(y-x)\nonumber \\
 &+\dl_{db}\,A_{\nu\, c}(y)_0\bigl(\del^0 D^{\nu\rho}_\la(y-x)-\del^\nu D^{0\rho}_\la(y-x)\bigr)\nonumber\\
&-\dl_{cb}\,A_d^0(y)_0\,\la\,\del_\mu D^{\mu\rho}_\la(y-x)-\dl_{db}\,(\del A_c)(y)_0\,\la\,D^{0\rho}_\la(y-x)\Bigr)\ .
\end{align}
Now we choose $y^0=x^0$ in order that we may use the equal-time commutation relations \eqref{eq:D-la(0,x)}.
Hence, solely the commutators $[\del^0 A(x^0,\vec y),A(x)]_\star\sim\del^0\,D_\la(0,\vec y-\vec x)\sim\dl(\vec y-\vec x)$ contribute.
By straightforward computation we obtain that the terms depending explicitly on $\la$ cancel out and that 
\be
\mathrm{\eqref{eq:[Q,A]}}=-i\hbar\ka\,f_{acb}\,A^\rho_c(x)_0=i\hbar\,\sS_a\bigl(A^\rho_b(x)\bigr)_0\ .
\ee
Proceeding analogously one verifies the assertion \eqref{eq:[Q,F]=S_a(F)} for $F=u(x)$ and $F=\tilde u(x)$; 
this verification is even somewhat simpler than for $F=A^\rho(x)$, because the free field equation is simpler. 

Since $\sS_a$ is a derivation w.r.t.~the classical product, more precisely
$\sS_a(F\cdot G)_0=\sS_a(F)_0\cdot G_0+F_0\cdot\sS(G)_0$, it remains to prove that 
\be\label{eq:[Q0,(n)]}
[Q_0,\vf^1_{b_1,0}(x_1)\cdot\ldots\cdot\vf^n_{b_n,0}(x_n)]_\star =i\hbar
\sum_{k=1}^n\vf^1_{b_1,0}(x_1)\cdot\ldots\cdot\sS_a\bigl(\vf^k_{b_k}(x_k)\bigr)_0\cdot\ldots\cdot\vf^n_{b_n,0}(x_n)\ ,
\ee
where $\vf^k\in\{A^\mu,\,u,\,\tilde u\}$. We proceed by induction on $n$. We have just verified the case $n=1$, i.e.,
$[Q_0,\vf_{b,0}(x)]_\star =i\hbar\,\sS_a\bigl(\vf_b(x)\bigr)_0=i\hbar\ka\, f_{abc}\,\vf_{c,0}(x)$. 
In the step $n\to n+1$ we use the relation
\footnote{\blue{This formula is obtained as follows: inserting the definition of the $\star$-product \eqref{eq:star-product}
into the first term term on the r.h.s., only the terms $1+\hbar\sD$ of $e^{\hbar\sD}$ contribute; $1$ gives the
term on the l.h.s.~and $\hbar\sD$ gives the second term on the r.h.s..}}
\begin{align}\label{eq:star(n+1)}
\vf^1_{b_1,0}&(x_1)\cdot\ldots\cdot\vf^{n+1}_{b_{n+1},0}(x_{n+1})=\bigl(\vf^1_{b_1,0}(x_1)\cdot\ldots\cdot\vf^n_{b_n,0}(x_n)
\bigr)\star \vf^{n+1}_{b_{n+1},0}(x_{n+1})\\
&-\sum_{k=1}^n\sgn(\pi_k)\,\bigl(\vf^1_{b_1,0}(x_1)\cdot\ldots\hat k\ldots\cdot\vf^n_{b_n,0}(x_n)\bigr)\cdot
\om_0\bigl(\vf^k_{b_k,0}(x_k)\star\vf^{n+1}_{b_{n+1},0}(x_{n+1})\bigr)\ ,\nonumber
\end{align}
where $\hat k$ means that $k$th factor is omitted and $\sgn(\pi_k)$ is the fermionic sign coming from the permutation
$\vf^1_{b_1,0}(x_1)\cdot\ldots\cdot\vf^n_{b_n,0}(x_n)\mapsto\vf^1_{b_1,0}(x_1)\cdot\ldots\hat k\ldots\cdot\vf^n_{b_n,0}(x_n)
\cdot\vf^k_{b_k,0}(x_k)$. We also take into account that the commutator $[Q_0,\bullet]_\star$ is a derivation w.r.t.~the 
on-shell $\star$-product. With that we obtain
\begin{align}\label{eq:[Q0,(n+1)]}
\frac{1}{i\hbar}[Q_0,\,&\vf^1_{b_1,0}(x_1)\cdot\ldots\cdot\vf^{n+1}_{b_{n+1},0}(x_{n+1})]_\star\nonumber\\
={}&\sum_{k=1}^n\bigl(\vf^1_{b_1,0}(x_1)\cdot\ldots\cdot\sS_a\bigl(\vf^k_{b_k}(x_k)\bigr)_0\cdot\ldots\cdot\vf^n_{b_n,0}(x_n)\bigr)
\star \vf^{n+1}_{b_{n+1},0}(x_{n+1})\nonumber\\
&+\bigl(\vf^1_{b_1,0}(x_1)\cdot\ldots\cdot\vf^n_{b_n,0}(x_n)\bigr)\star \sS_a\bigl(\vf^{n+1}_{b_{n+1}}(x_{n+1})\bigr)_0
\nonumber\\
&-\sum_{k=1}^n\sgn(\pi_k)\,\sum_{j(\not= k)}\bigl(\vf^1_{b_1,0}(x_1)\cdot\ldots\hat k\ldots
\cdot\sS_a\bigl(\vf^j_{b_j}(x_j)\bigr)_0\cdot\ldots\cdot\vf^n_{b_n,0}(x_n)\bigr)\nonumber\\
&\qquad\qquad\qquad\qquad\cdot\om_0\bigl(\vf^k_{b_k,0}(x_k)\star\vf^{n+1}_{b_{n+1},0}(x_{n+1})\bigr)\ .
\end{align}
Since it holds that
\begin{align*}
\om_0&\Bigl(\sS_a\bigl(\vf^k_{b_k}(x_k)\bigr)_0\star\vf^{n+1}_{b_{n+1}}(x_{n+1})_0\Bigr)+
\om_0\Bigl(\vf^k_{b_k}(x_k)_0\star\sS_a\bigl(\vf^{n+1}_{b_{n+1}}(x_{n+1})\bigr)_0\Bigr)\\
&\overset{\eqref{eq:Sa-starproduct}}{=}\om_0\Bigl(\sS_a\bigl(\vf^k_{b_k}(x_k)\star\vf^{n+1}_{b_{n+1}}(x_{n+1})\bigr)_0\Bigr)=0
\end{align*}
(by using $\om_0\circ\sS_a=0$), we may add
\begin{align*}
0={}&-\sum_{k=1}^n\sgn(\pi_k)\,\bigl(\vf^1_{b_1,0}(x_1)\cdot\ldots\hat k\ldots\cdot\vf^n_{b_n,0}(x_n)\bigr)\\
&\cdot\Bigl(\om_0\bigl(\sS_a(\vf^k_{b_k}(x_k))_0\star\vf^{n+1}_{b_{n+1}}(x_{n+1})_0\bigr)+
\om_0\bigl(\vf^k_{b_k}(x_k)_0\star\sS_a(\vf^{n+1}_{b_{n+1}}(x_{n+1}))_0\bigr)\Bigr)
\end{align*}
to the right hand side of \eqref{eq:[Q0,(n+1)]}; using again \eqref{eq:star(n+1)} we see that 
 the r.h.s~of \eqref{eq:[Q0,(n+1)]} is equal to the r.h.s.~of the assertion \eqref{eq:[Q0,(n)]} with $(n+1)$ in place of $n$.
$\qquad\qquad\qquad\qed$

\begin{remk}
The situation is similar to spinor QED (see \cite{DF99} or \cite[Chap.~5.2]{D19}) and scalar QED (see \cite{DPR21}): In these models the 
\emph{global} transformation is the \emph{charge number operator},
\be\label{eq:theta}
\th :=\int dx\,\,\dl_{Q(x)}\word{with}\dl_{Q(x)}F:=\Bigl(\frac{\dl_r F}{\dl\psi(x)}\wedge \psi(x)-
\overline \psi(x)\wedge\frac{\dl F}{\dl\overline\psi(x)}\Bigr)\ ,\quad F\in\sF\ ,
\ee
for spinor QED and $\dl_{Q(x)}:=\phi(x)\,\fd{}{x}-\phi^*(x)\,\frac{\dl}{\dl\phi^*(x)}$ for scalar QED.
Proceeding similarly to \eqref{eq:Noether-Sa}-\eqref{eq:current-Sa}, we see that the classical 
Noether current $j^\mu(g;x)$ belonging to charge number conservation of the total Lagrangian, i.e., 
$\th\bigl(L_0(x)+L_\inter(g;x))\bigr)=0$, satisfies also the relation \eqref{eq:dj} with $\th$ in place of $\sS_a$,
\footnote{Note that for both spinor and scalar QED, the sum over $k$ in \eqref{eq:dj} runs effectively only over
$\vf_k=\psi,\ovl\psi$ or $\vf_k=\phi,\phi^*$, respectively, because $\th\bigl(A^\mu(x)\bigr)=0$. This is in accordance
with the just given definition of $\dl_{Q(x)}$.}
\blue{explicitly
\be
\del_\mu j^\mu (g;x)=i\Bigl( \frac{\dl_r S_0}{\dl\psi(x)}\wedge \psi(x)
-\ovl\psi(x)\wedge\frac{\dl S_0}{\dl\overline\psi(x)}\Bigr)\word{with} j^\mu(g;x)=\ovl\psi(x)\wedge\ga^\mu\psi(x)\ ,
\ee
and $S_0:=\int dy\,\,\ovl\psi(y)\wedge (i\slashed\del_y -m\,1_{4\x 4})\psi(y)+S_\photon(A)+S_\gf(A)$
for spinor QED, see the explanations concerning the current $j^\mu(g;x)$ given at the end of this remark.}
Hence, for the on-shell AMWI belonging to the 
corresponding \emph{local} transformation, i.e.,
\be
\dl_{qQ}:=\int dx\,\,q(x)\,\dl_{Q(x)}\word{with $q\in\sD(\bM)$ arbitrary,}
\ee
we obtain also the \emph{anomalous current conservation} \eqref{eq:AMWI-j-onshell}, \blue{where,
similarly to \eqref{eq:AMWI-j-onshell}, the test function $g$ switching the interaction is arbitrary.} 
Proceeding analogously to the proof of
Prop.~\ref{pr:int-anomaly}, it is  proved in these references that invariance of the $R$-product w.r.t. $\th$ (i.e., \eqref{eq:[Sa,R]=0}
with $\th$ in place of $\sS_a$)
\footnote{The proof of that version of Prop.~\ref{pr:[Sa,R]=0} is simpler than for $\sS_a$, because all field monomials are eigenvectors of $\th$ with discrete eigenvalues, see the cited references.}
 $\,$ implies that the integral over the anomaly (of the current conservation) 
vanishes also in these models, i.e., \eqref{eq:int-dj(g;x)=0} holds also there. 
However, to prove that the anomaly itself can be removed by finite, admissible renormalizations 
\blue{(without assuming $x\in g^{-1}(1)^\circ$)}
requires quite a lot of additional work, including an investigation of some classes of Feynman diagrams. For the Noether current 
$j^\mu_{S_\inter,0}(x)\equiv j^\mu_{S_\inter,0}(g;x)$
belonging to $\sS_a$ and for the non-Abelian gauge current $J^\mu_{S_\inter,0}(x)$, both restricted to the region 
$x\in g^{-1}(1)^\circ$ (see \eqref{eq:dj=dJ}), we have given a shorter and more elegant proof of the removability of the anomaly 
$(\Delta(x)(S_\inter))_{S_\inter,0}$ (in Sect.~\ref{sec:ccAMWI}) by using the consistency condition for the (possible) anomaly.
\blue{Unfortunately, the anomaly consistency condition \eqref{eq:WZ} is trivial (i.e., $0=0$) for spinor and scalar QED.}

Note that for scalar QED the Noether current of the interacting theory is%
\footnote{We use here \eqref{eq:current-Sa} and that $L(g;x)=(D^\mu\phi)^*(x)\,(D^\mu\phi)(x)-m^2\,\phi^*(x)\phi(x)+
\blue{L_\photon(x)+L_\gf(x)}$. 
\blue{(Note that $L_\photon$ does not contain any interaction term and, hence, does not depend on $g$, similarly to $L_\gf$.)}}
\be
j^\mu(g;x)=i\bigl(\phi(x)\,(D^\mu\phi)^*(x)-\phi^*(x)\,D^\mu\phi(x)\bigr)\word{with}D^\mu_x:=\del^\mu_x+i\ka g(x)\,A^\mu(x)
\ee
being the covariant derivative; hence, similarly to the situation in this paper, $j^\mu(g;x)$ differs from the corresponding current of the 
free theory $ j^\mu(0;x)$. But in spinor QED there is the peculiarity that these two currents agree:
$j^\mu(g;x)=j^\mu(0;x)=\ovl\psi(x)\wedge\ga^\mu\psi(x)$, because $L_\inter(x)=j^\mu(0;x)\,A^\mu(x)$ satisfies
$\frac{\del L_\inter}{\del(\del_\mu \psi)}=0=\frac{\del L_\inter}{\del(\del_\mu \ovl\psi)}$, see \eqref{eq:current-Sa}.
\end{remk}

 \medskip

\paragraph{Acknowledgements.} The author \blue{thanks the referee for reading the manuscript very thoroughly and 
pointing out a lot of improvements. He also}
profited a lot from stimulating discussions with Romeo Brunetti, Klaus Fredenhagen
and Kasia Rejzner. In particular, Section \ref{sec:global-trafo} was initiated by insights of Klaus Fredenhagen.
\medskip

\appendix
\section{\blue{Proof of $\langle(\Dl X)'(F),\del^r\vf(x)\rangle=0$ (Thm.~\ref{thm:AMWI}(b6)), and  the Action Ward Identity
for $\Dl^n(\bullet;X)$}}

\blue{We start with  the proof of Thm.~\ref{thm:AMWI}(b6):
\begin{proof}
To simplify the notations we study a model with only \emph{one} basic field $\vf$ with $\fd{S_0}{x}=-(\square+m^2)\vf(x)$
and use that $\del_X=\int dy\,\,\del_X\vf(y)\,\fd{}{y}$
(in accordance with \eqref{eq:del_X}). The proof of Thm.~\ref{thm:AMWI}(a) proceeds by an inductive construction of the 
sequence $(\Dl^n)_{n\in\bN}$, hence we argue here also by induction on $n$. 
The inductive relation follows from the AMWI \eqref{eq:AMWI}, it is given in \eqref{eq:recursion-Dln},
(or \cite[formula (5.15)]{Brennecke08}) or \cite[formula (4.3.10)]{D19}); in the case
 relevant for this proof it reads
 \begin{align}\label{eq:Dl(vf;X)}
\hbar^n\,&\Dl^n\Bigl(\vf(x)\ox F^{\ox (n-1)};X\Bigr)=
-\int dy\,\,R_{n,1}\Bigl(\vf(x)\ox F^{\ox (n-1)},\del_X\vf(y)\Bigr)\cdot\frac{\dl S_0}{\dl\vf(y)}\\
&+R_{n,1}\Bigl(\vf(x)\ox F^{\ox (n-1)},\del_X S_0\Bigr)\nonumber\\
&+\hbar\Bigl[(n-1)\,R_{n-1,1}\Bigl(\vf(x)\ox F^{\ox (n-2)},\del_X F\Bigr)
+R_{n-1,1}\Bigl(F^{\ox (n-1)},\del_X\vf(x)\Bigr)\Bigr]\nonumber\\
&-\sum_{l=0}^{n-1}\hbar^l\, \binom{n-1}{l}\,
R_{n-l,1}\Bigl(\vf(x)\ox F^{\ox (n-1-l)},\Dl^l\bigl(F^{\ox l};X\bigr)\Bigr)\ ;\nonumber
\end{align}
in the last term we have taken into account the inductive assumptiom:
$\Dl^l\bigl(\vf(x)\ox\cdots;X\bigr)=0$ for $l<n$.
To show that the r.h.s.~vanishes, we use that as a consequence of the renormalization conditions
Off-shell field equation and Field independence it holds that 
\be
R_{n,1}\Bigl(\vf(x)\ox F^{\ox (n-1)},G\Bigr)=-\hbar\int dz\,\,\Dl^\reta(z-x)\,\fd{}{z}R_{n-1,1}\Bigl(F^{\ox (n-1)},G\Bigr)
\ee
(see \cite[Lem.~1(B)]{DF99} or \cite[formula (3.2.61)]{D19});
and, for the first term we use the relation
\begin{align*}
\fd{S_0}{y}\cdot\fd{}{z}\,R_{n-1,1}\Bigl(F^{\ox (n-1)},\del_X\vf(y)\Bigr)={}&
\fd{}{z}\Bigl\{\fd{S_0}{y}\cdot R_{n-1,1}\Bigl(F^{\ox (n-1)},\del_X\vf(y)\Bigr)\Bigr\}\\
+(\square_z+m^2)&\dl(z-y)\,\,R_{n-1,1}\Bigl(F^{\ox (n-1)},\del_X\vf(y)\Bigr)\ .
\end{align*}
Proceeding in this way, the assertion reads
\begin{align}\label{eq:Dl(vf)=0}
0\overset{?}{=}{}
&-\int dz\,\,\Dl^\reta(z-x)\,\fd{}{z}\Bigl\{
-\int dy\,\,R_{n-1,1}\Bigl( F^{\ox (n-1)},\del_X\vf(y)\Bigr)\cdot\frac{\dl S_0}{\dl\vf(y)}\\
&+R_{n-1,1}\Bigl(F^{\ox (n-1)},\del_X S_0\Bigr)+\hbar\, (n-1)\,R_{n-2,1}\Bigl(F^{\ox (n-2)},\del_X F\Bigr)\nonumber\\
&-\sum_{l=0}^{n-1}\hbar^l\, \binom{n-1}{l}\,
R_{n-1-l,1}\Bigl(F^{\ox (n-1-l)},\Dl^l\bigl(F^{\ox l};X\bigr)\Bigr)\Bigr\}\nonumber\\
&+\int dy\,\Bigl\{\int dz\,\,\Dl^\reta(z-x)\,(\square_z+m^2)\dl(z-y)+
\dl(x-y)\Bigr\}\,R_{n-1,1}\Bigl(F^{\ox (n-1)},\del_X\vf(y)\Bigr)\ ,\nonumber
\end{align}
where an overall factor $\hbar$ is omitted.
Both $\{\cdots\}$-brackets vanish; for the first one this is due to
the anomalous MWI to order $(n-1)$ and for the second one this is due to $(\square+m^2)\Dl^\reta=-\dl$.
Finally, the generalisation of the assertion to $\langle(\Dl X)'(F),\del^r\vf(x)\rangle=0,\,\,r\in\bN_0^4$, follows from the validity of the 
AWI for $\Dl^n$, see the following Remark.
\end{proof}}

\blue{\begin{remk}[Action Ward Identity for $\Dl^n$]\label{rem:AWI-Dl}
The statement that the maps $\Dl^n$ depend only on local functionals (see \eqref{eq:anom-terms}) implies the AWI for $\Dl^n$:
\be\label{eq:Dl-AWI}
\Dl^n\bigl((\del^\mu P)(x) \ox F^{\ox (n-1)};X\bigr)=\del^\mu_x\Dl^n\bigl(P(x) \ox F^{\ox (n-1)};X\bigr)\word{for} P\in\sP\ .
\ee
Hence, the inductive construction of the sequence $(\Dl^n)_{n\in\bN}$ has to be such that this identity holds.
This holds indeed, due to the facts that $\Dl^n$ is constructed in terms of $R$-products and the latter fulfil the AWI. In detail,
the inductive relation for $\Dl^n\bigl((\del^\mu P)(x) \ox F^{\ox (n-1)};X\bigr)$ is obtained from \eqref{eq:Dl(vf;X)} by replacing 
throughout $\vf(x)$ by $\del^\mu P(x)$ (more precisely: 
on the r.h.s., in the first term $\del_X\vf(y)$ remains as it is, but in the second last term
$\del_X\vf(x)$ is replaced by $\del_X (\del^\mu P)(x)=\del_x^\mu(\del_X P(x))$, cf.~\eqref{eq:dXdA=ddXA}) and by adding 
on the r.h.s.~the term
$$
-\sum_{l=0}^{n-2}\hbar^{l+1}\, \binom{n-1}{l}\,
R_{n-1-l,1}\Bigl( F^{\ox (n-1-l)},\Dl^{l+1}\bigl((\del^\mu P)(x)\ox F^{\ox l};X\bigr)\Bigr)\ ,
$$
cf.~\eqref{eq:recursion-Dln}.
The assertion \eqref{eq:Dl-AWI} follows then by using the AWI for $\Dl^{l+1}(\bullet;X)$ and for the $R$-products.
\end{remk}}

\end{document}